\documentclass[journal]{IEEEtran}

\usepackage{siunitx}
\usepackage{booktabs}

\usepackage{tabularx}
\usepackage[T1]{fontenc}
\usepackage{cite}
\usepackage{amsthm}
\usepackage{amsmath}
\usepackage{algorithm,algpseudocode}
\usepackage[dvips]{graphicx}
\usepackage[font=small]{caption}
\usepackage{subfig}
\usepackage{float}

\usepackage{amssymb}
\usepackage{multirow}
\usepackage{colortbl}
\usepackage{array, tabularx}
\usepackage{color}
\usepackage{amsbsy}
\usepackage{bm}
\usepackage{fixmath}
\usepackage{mathtools}
\DeclarePairedDelimiter{\ceil}{\lceil}{\rceil}
\usepackage{dsfont}
\usepackage{bm}
\usepackage{setspace}
\usepackage{ctable}

\ifCLASSINFOpdf
\else
\fi
\begin{document}
%
\title{\onehalfspacing
Pilot Reuse Strategy Maximizing the Weighted-Sum-Rate in Massive MIMO Systems}



\author{Jy-yong Sohn,
        Sung Whan Yoon,~\IEEEmembership{Student Member,~IEEE,}
        and~Jaekyun~Moon,~\IEEEmembership{Fellow,~IEEE}
\thanks{Manuscript received December 1, 2016; accepted March 6, 2017. The authors are with the School
of Electrical Engineering, Korea Advanced Institute of Science and Technology, Daejeon, 305-701, Republic of Korea (e-mail: jysohn1108@kaist.ac.kr, shyoon8@kaist.ac.kr, jmoon@kaist.edu).}
}


%

\markboth{to appear in IEEE Journal on Selected Areas in Communications 2017}%
{Shell \MakeLowercase{\textit{et al.}}: Bare Demo of IEEEtran.cls for Journals}


\maketitle

\begin{abstract}
Pilot reuse in multi-cell massive multi-input multi-output (MIMO) system is investigated where user groups with different priorities exist.
Recent investigation on pilot reuse has revealed that 
when the ratio of the coherent time interval to the number of users
is reasonably high, it is beneficial not to fully reuse pilots from 
interfering cells. This work finds the optimum pilot assignment strategy
that would maximize the weighted sum rate (WSR) given the user groups 
with different priorities. 
A closed-form solution for the optimal pilot assignment is derived
and is shown to make intuitive sense.
Performance comparison shows that under wide range of channel conditions, the optimal pilot assignment that uses extra set of pilots  
achieves better WSR performance than conventional full pilot reuse.
\end{abstract}

\begin{IEEEkeywords}
Massive MIMO, Multi-cell MIMO, Pilot contamination, Pilot assignment, Pilot reuse, Channel estimation, Weighted-Sum-Rate maximization
\end{IEEEkeywords}

%
\IEEEpeerreviewmaketitle

\section{Introduction}

Deployment of multiple antennas at the transmitter and the receiver, collectively known as MIMO technology, has been 
instrumental in improving link reliability as well as throughput of modern wireless communication systems.
Under the multi-user MIMO setting, the latest development has been the use of 
an exceedingly large number of antennas at base stations (BSs) compared to the number of user terminals (UTs) served by each BS. Under this "massive" MIMO setup, assuming time-division-duplex (TDD) operation with uplink pilot training for 
channel state information (CSI) acquisition, the effect of fast-fading coefficients and uncorrelated noise disappear as the number of BS antennas increases without bound \cite{marzetta2006much, marzetta2010noncooperative}. 
Massive MIMO is considered as a promising technique in 5G communication systems, which has a potential of increasing spectral and energy efficiency significantly with simple signal processing \cite{lu2014overview, larsson2014massive,andrews2014will}. 
The only factor that limits the achievable rate of massive MIMO system is the pilot contamination, which arises due to the reuse of the same pilot set among interfering cells. The pilot reuse causes the channel estimator error at each BS, where precoding scheme is also contaminated by inaccurate channel estimation. This phenomenon degrades the achievable rate of a massive MIMO system even when $M$, the number of BS antennas, tends to infinity; this remains as a fundamental issue in realizing massive MIMO. 

Various researchers have suggested ways to mitigate pilot contamination effect, which are well summarized in \cite{elijah2015comprehensive}.
A pilot transmission protocol which reduces pilot contamination is suggested in \cite{appaiah2010pilot}. They considered a time-shifted pilot transmission, which systematically avoids collision between non-orthogonal pilot signals. However, central control is required for this scheme, and back-haul overload remains as an issue.
Realizing that inaccurate channel estimation is the fundamental reason for pilot contamination, some researchers have focused on effective channel estimation methods for reducing pilot contamination.  
Exploitation of the angle-of-arrival (AoA) information in channel estimation is considered 
\cite{yin2013coordinated}, which achieves interference-free channel estimation and effectively eliminates pilot contamination when the number of antennas increases without bound. However, the channel estimation method of \cite{yin2013coordinated} uses $2^{nd}$ order statistics for the channels associated with pilot-sharing users in different cells, requiring inter-cell cooperation via a back-haul network.
Under the assumption of imperfect CSI, another researchers investigated uplink/downlink signal processing schemes to mitigate pilot contamination \cite{jose2011pilot, ashikhmin2012pilot, li2015multi}. A 
precoding method called pilot contamination precoding (PCP) is suggested  to reduce pilot contamination \cite{ashikhmin2012pilot}. This method utilizes slow fading coefficient information for entire cells, and cooperation between BSs is required. Multi-cell MMSE detectors are considered in several literatures \cite{ngo2012performance, guo2013performance, guo2014uplink, li2015multi}, but interference due to pilot contamination cannot be perfectly eliminated.
Moreover, most of these works rarely consider the potential of pilot allocation in reducing pilot contamination.

Some recent works shed light on appropriate pilot allocation as a candidate to tackle pilot contamination. 
There are three types of pilot allocation: one is to allocate orthogonal pilots within a cell assuming full pilot reuse among different cells \cite{yin2013coordinated, zhu2015smart, nguyen2015resource}, another is reusing a single pilot sequence among users within a cell, while users in different cells employ orthogonal pilots \cite{liu2015pilot}, and the other is finding the pilot reuse rule among different cells whereas all users within a cell are guaranteed to be assigned orthogonal pilots \cite{li2015multi, bjornson2016massive, saxena2015mitigating, sohn2015pilots}.

A coordination-based pilot allocation rule is suggested in \cite{yin2013coordinated}, which shapes the covariance matrix in order to achieve interference-free channel estimation and eliminate pilot contamination. However, back-haul overload is required to apply the suggested pilot allocation method. Another coordination-based pilot allocation in \cite{zhu2015smart} utilizes slow-fading coefficient information, where the pilot with the least inter-cell interference is assigned to the user with the worst channel quality. The optimal allocation rule of resources (transmit power, the number of BS antennas, and the pilot training signal) which maximizes the spectral efficiency is considered in \cite{nguyen2015resource}. Cooperation among BSs is assumed in this optimization, where each BS transmits the path-loss coefficients information to other interfering cells.
All these works show that appropriate pilot assignment is important in reducing pilot contamination, but they all suggest coordination-based solutions which have implementation issues. 
On the other hand, reusing the same pilot within a cell is considered in \cite{liu2015pilot}, whereas interfering cells have orthogonal pilots. The pilot reuse within a cell causes intra-cell interference, which is eliminated by the downlink precoding scheme suggested in the same paper. However, the interference is fully eliminated only when the number of BS antennas is infinite and the equivalent channel is invertible; the solution for a general setting still remains as an open problem.

Unlike conventional full pilot reuse, some researchers considered less aggressive pilot reuse scheme as a candidate for reducing pilot contamination. Simulation results in \cite{li2015multi, bjornson2016massive, saxena2015mitigating} showed that less aggressive pilot reuse can increase spectral efficiency in some practical scenarios, but the approaches are not based on a closed-form solution, which does not offer useful insights into the trade-off between increased channel estimation accuracy and decreased data transmission time.
Perhaps \cite{sohn2015pilots} is the first paper which observed and mathematically analyzed the trade-off when less aggressive pilot reuse is applied. In \cite{sohn2015pilots}, 
the present authors 
analyzed the potential of using more pilots than $K$, the number of UTs in a BS, to mitigate pilot contamination and increase the achievable rate. Based on lattice-based hexagonal cell partitioning, they 
formulated the relationship between normalized coherence time and the optimal number of orthogonal pilots utilized in the system. Also, the optimal way of assigning the pilots to different users are specified in a closed-form solution. The optimality criterion was to maximize the net sum-rate, which is the sum rate counting only the data transmission portion of the coherence time. It turns out that for many practically meaningful channel scenarios, departing from conventional full pilot reuse and selecting optimal pilot assignment increases the net sum-rate considerably.


This paper expands the concept of optimal pilot assignment to the practical scenario where different user groups exist with different priorities and where it is necessary to maximize the net weighted-sum-rate (WSR).
A practical communication system which guarantees
sufficient data rates for high-paying selected customers
can be considered.
We formulate a WSR maximization problem by prioritizing the users into several groups and giving different weights to different groups. The higher priority group would get a higher weight, so that the achievable rate of the higher priority group has a bigger effect on the objective function.
We grouped the users by their data rate requirements; in the Internet-of-Things (IoT) era where many different types of devices participate in the network, this kind of differentiation might be helpful. For example, small sensors with less throughput requirement can be considered as lower-priority users, whereas devices with large amount of computation/communication can be considered as higher-priority users.
 In this paper, we consider two priority groups: preferred user group ($1^{st}$ priority group) and regular user group ($2^{nd}$ priority group), but similar result are expected in general multiple priority groups case, as suggested in Section \ref{Section:Further Comments}. A closed-form solution for optimal assignment is obtained, which is consistent with 
intuition. Compared to the result of \cite{sohn2015pilots}, the net-WSR of the optimal assignment beats conventional assignment for a wider range of channel coherence time, which means that departing from full pilot reuse and applying optimal less aggressive pilot reuse is necessary particularly in practical net-WSR maximizing scenarios.

This paper is organized as follows. Section \ref{Section:SystemModel} describes the system model for massive MIMO and the 
pilot contamination effect.
Section \ref{Section:Preliminaries} summarizes the pilot assignment strategy for multi-cell massive MIMO of \cite{sohn2015pilots}, which acts as preliminaries for the main analysis of this paper. Section \ref{Section:OptimizationProblem} formulates the net-WSR-maximizing pilot assignment problem and the closed-form solution is suggested in Section \ref{Section:OptimizationSolution}. Simulation results are included in Section \ref{Section:Simulation}, which support the mathematical results of Section \ref{Section:OptimizationSolution}. Here, comparison is made between the performances of the optimal and conventional assignments. In Section \ref{Section:Further Comments}, further comments on the scenarios with multiple (three or more) priority groups and finite BS antennas are given.
Finally, Section \ref{Section:Conclusion} concludes the paper.

\section{System Model}\label{Section:SystemModel}

\subsection{Multi-Cellular Massive MIMO System}

Consider a communication network with $L$ hexagonal cells, where each cell has $K$ single-antenna users located in a uniform-random manner. Each BS with multiple antennas estimates downlink CSIs by uplink pilot training, assuming channel reciprocity in TDD operation.
The channel model used in this paper is assumed to be identical to that in \cite{bjornson2016massive}; this model fits well with the real-world simulation tested by \cite{gao2015massive}, for both few and many BS antennas. 
Two types of channel models, independent channel and spatially correlated channel, are used in massive MIMO as stated in \cite{elijah2015comprehensive}, while this paper assumes an independent channel model. Antenna elements are assumed to be uncorrelated in this model, which is reasonable with sufficient antenna spacing. 
The complex propagation coefficient $g$ of a link is decomposed into a complex fast-fading factor $h$ and a slow-fading factor $\beta$. The channel link between the $m^{th}$ BS antenna of $j^{th}$ cell and the $k^{th}$ user in the $l^{th}$ cell is modeled as 
$g_{mjkl} = h_{mjkl} \sqrt{\beta_{jkl}}.$
Here, the fast-fading factor $h_{mjkl}$ of each link is modeled as an independent and identically distributed (i.i.d.) complex Gaussian random variable with zero-mean and unit variance.
The slow-fading factor $\beta$ is modeled as
$\beta_{jkl} = {1}/{r_{jkl}^\gamma}$
where $\gamma$ is the signal decay exponent ranging from $2$ to $4$, and 
$r_{jkl}$ is the distance between the $k^{th}$ user of the $l^{th}$ cell and the BS of $j^{th}$ cell.

The channel coherence time and channel coherence bandwidth are denoted as $T_{coh}$ and $B_{coh}$, respectively, while $T_{del} = 1/B_{coh}$ represents the channel delay spread.
Here, we define a dimensionless quantity, \textit{normalized coherence time} $N_{coh}=T_{coh}/T_{del}=T_{coh}B_{coh}$, to represent the number of independently usable time-slots available within the coherence time. 
For a specific numerical example, consider the practical scenario based on OFDM with a frequency smoothness interval of $N_{smooth} = 14$ (i.e., the fast fading coefficient is constant for 14 successive sub-carriers), and a coherence time ranging $20$ OFDM symbols. This example has the corresponding normalized coherence time of $N_{coh}=280$.
The normalized coherence time is divided into two parts: pilot training and data transmission. Based on the CSI estimated in the pilot training phase, data is transmitted in the rest of the coherence time.

\subsection{Pilot Contamination Effect}

In a massive MIMO system with TDD operation, $K$ users in each cell are usually assumed to use orthogonal pilots, so that BS can estimate the channel to each user by collecting uplink pilot signals without interference. However, in the multi-cell system, due to a finite $N_{coh}$ value, it is hard to guarantee orthogonality of pilot sets among adjacent cells. 
Therefore, users in different cells might have non-orthogonal pilot signals, which contaminates the channel estimates of the users. This effect is called the pilot contamination effect, which saturates the achievable rate even as $M$, the number of BS antennas, increases without bound.
The saturation value can be expressed as follows.
Assume each cell has a single user, where the identical pilot signal is reused among different cells. Then, the uplink achievable rate of the user in the $j^{th}$ cell is saturated to
\begin{equation} \label{achievableR}
\lim_{M \to \infty} R_{j} = \log_{2} \left(1+\frac{\beta_{jj}^{2}}{\sum_{l\neq j}\beta_{jl}^{2}}\right)
\end{equation}
where $\beta_{jl}$ represents the slow-fading term of the channel between $j^{th}$ BS and the pilot-sharing user in the $l^{th}$ cell.


\section{Preliminaries}\label{Section:Preliminaries}

In this section, some preliminaries for pilot assignment strategy for multi-cell massive MIMO systems are given. Specifically, our previous work asserts that in some practical circumstances, optimized pilot assignment in a multi-cell massive MIMO system gives much improved sum rate performance than conventional full pilot reuse \cite{sohn2015pilots}. 

\subsection{Hexagonal-Lattice-Based Cell Clustering}

First, the locations of different users are assumed to be independent within a cell, and $K$ users in a given cell have orthogonal pilots. Thus the pilot assignment on $L$ cells with $K$ users each can be decomposed into $K$ independent sub-assignments on $L$ cells with single user each.


\begin{figure}
	\centering
	\subfloat[][3-way Partitioning]{\includegraphics[height =24mm ]{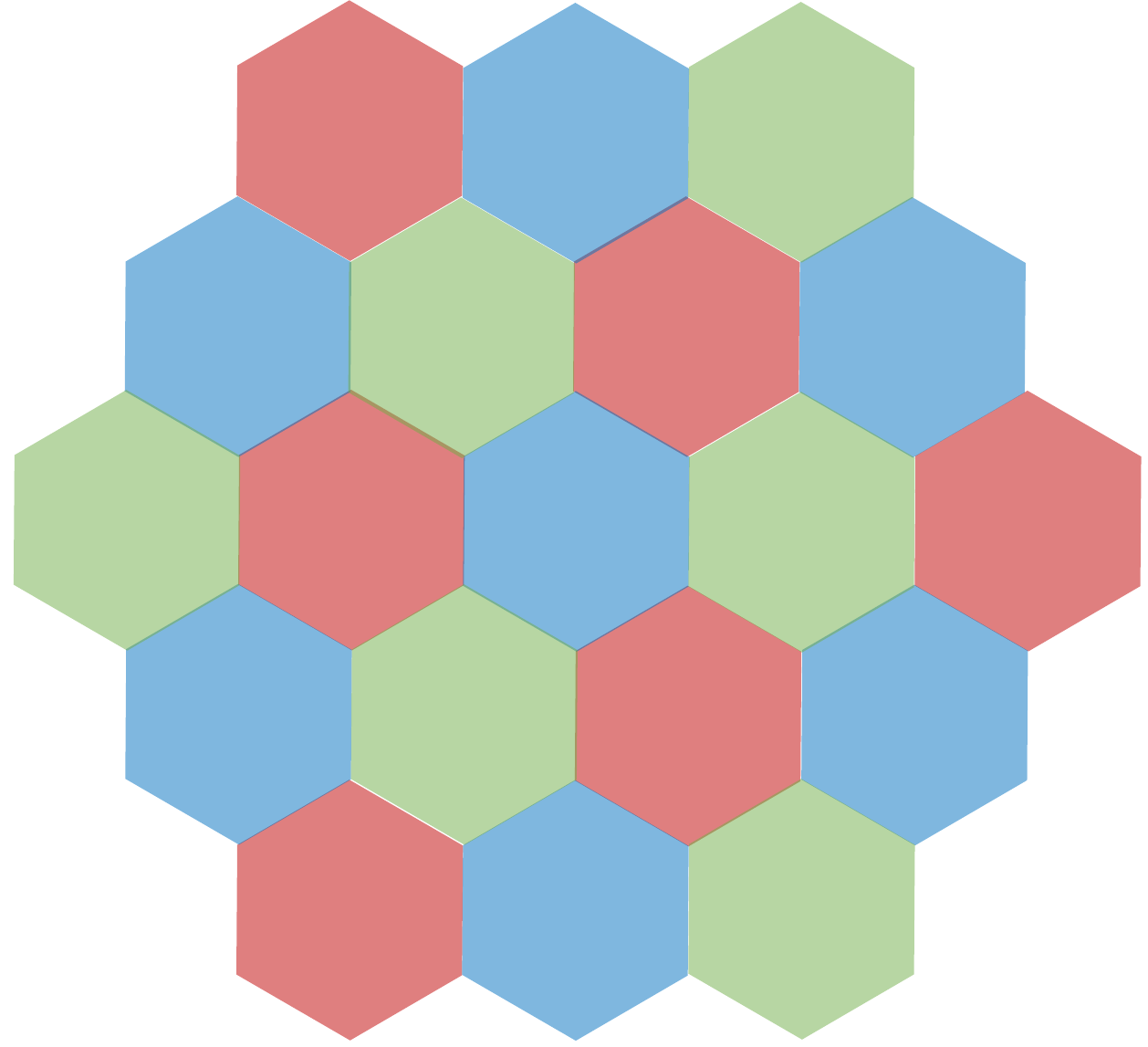}\label{Fig:3-way partitioning}}
	\quad \quad
	\subfloat[][Hierarchical set partitinoing]{\includegraphics[width=0.37\textwidth]{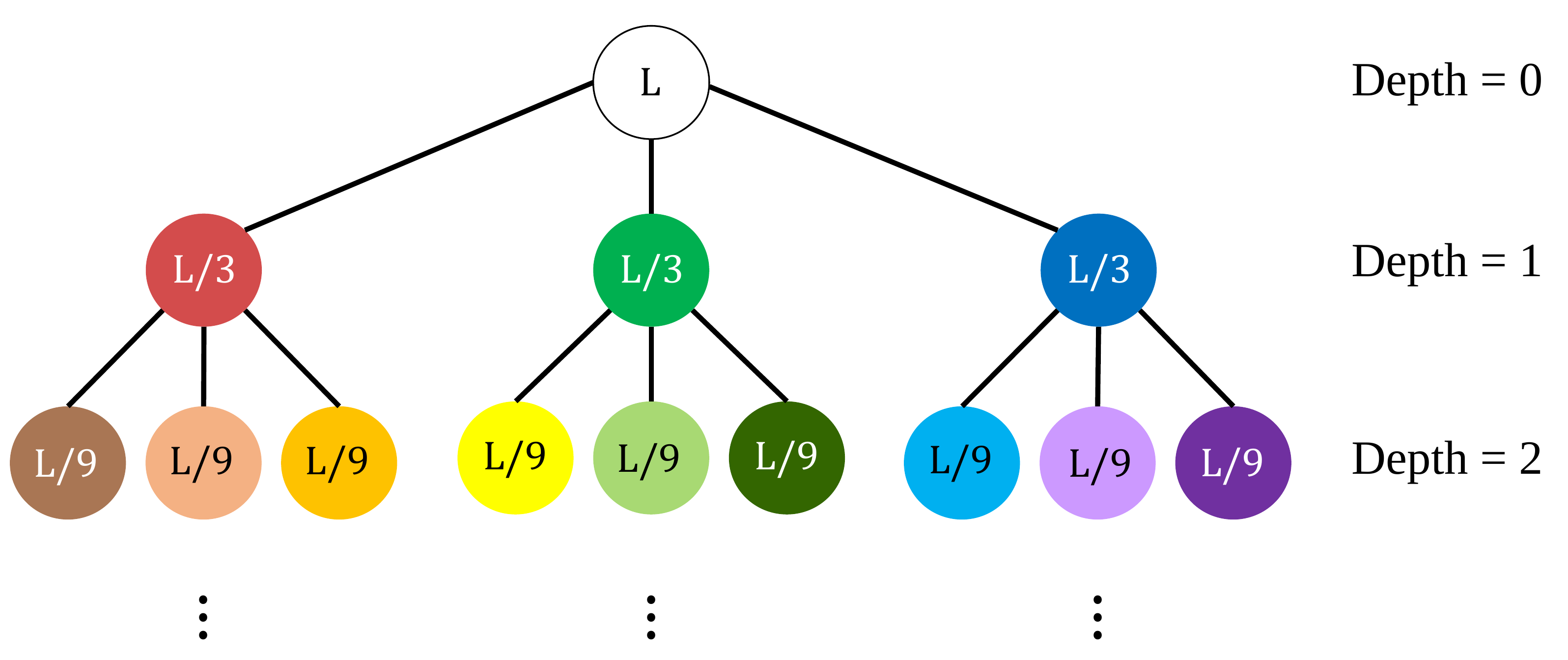}\label{Fig:Hierarchy tree structure}}
	\caption{The Cell Partitioning Method}
	\label{steady_state}
\end{figure}

%

For identifying which cells reuse the same pilot set,
consider an example of hexagonal cells with lattice structure in Fig. \ref{Fig:3-way partitioning}.  The 3-way partitioning groups the 19 cells into three equi-distance subsets colored by red, green and blue.
Since each subset in Fig. \ref{Fig:3-way partitioning} forms a lattice, we can consider applying this 3-way partitioning in a hierarchical manner as illustrated in Fig. \ref{Fig:Hierarchy tree structure}. 
In the tree structure, the root node represents $L$ cells with three children nodes produced by 3-way partitioning. 
After $i$ consecutive partitioning, depth $i$ can contain $3^{i}$ leaves where each leaf represents $L\cdot3^{-i}$ cells sharing the same pilot set. 
With these hierarchical partitioning and tree-like representation, pilot assignment on a multi-cell network can be uniquely represented.
The maximally achievable depth of the tree is set to $(\log_3 L - 1)$. This is because the number of pilot-sharing cells need to be greater than $1$, since the user who monopolizes a pilot has an infinite achievable rate (from (\ref{achievableR})), a meaningless situation. For meaningful analysis, we considered $L$ values with $(\log_3 L - 1) \geq 1$, i.e., $L \geq 9$.

Another importance associated with identifying the depth of a leaf is that the achievable rate of users in the set of cells corresponding to the leaf depends only on the depth. Let $C_{i}$ be the achievable rate of a user in a cell at depth $i$. $C_{i}$ is an increasing function of $i$ with a nearly linear behavior ($C_{i+1}-C_{i} \simeq 6$) \cite{sohn2015pilots}.

\newtheorem{theorem}{Theorem}
\newtheorem{lemma}{Lemma}
\newtheorem{corollary}{Corollary}
\newtheorem{prop}{Proposition}
\newenvironment{definition}[1][Definition]{\begin{trivlist} \item[\hskip \labelsep \normalfont #1]}{\end{trivlist}}

\subsection{Pilot Assignment Vector and Net Sum Rate}
The tree-like hierarchical representation of pilot assignment strategy can be uniquely converted into a vector form. Let $\mathbf{p}=(p_{0},p_{1},...,p_{\log_{3}L-1})$ be a vector where $p_{i}$ is the number of leaves in the corresponding tree-like hierarchical representation.

\begin{definition}[Definition:\nopunct]
Let $L, K$ be positive integers. For the given $L$ cells with $K$ users each, the set $P_{L,K}$ of valid pilot assignment vectors for $LK$ users based on 3-way partitioning is given by
\begin{align*}
P_{L,K} = \{\mathbf{p}&=(p_0,p_1,\cdots,p_{\log_3 L -1}): \\
\: &p_i \in \{0, 1, \cdots, K3^i \} \textrm{, and} \: \sum\limits_{i=0}^{\log_3 L -1} p_i3^{-i}=K \}
\end{align*} 
\end{definition}
For a pilot assignment vector $\mathbf{p}=(p_{0},p_{1},...,p_{\log_{3}L-1})$, the pilot length is defined as
$N_{pil}(\mathbf{p}) = \sum\limits_{i=0}^{\log_3 L-1} p_i $, which represents the number of pilots utilized in the system.
Because each element $p_{i}$ represents the number of leaves at depth $i$, a total of $L3^{-i}p_{i}$ users are located at depth $i$.
Recalling $C_{i}$ is the achievable rate of a user in the $i^{th}$ depth, the per-cell sum rate of the $L$-cell network with the pilot assignment scheme $\mathbf{p}$ is
\begin{equation}\label{sum rate}
C_{sum}(\mathbf{p})=\frac{1}{L}\sum\limits_{i=0}^{\log_3 L-1} L3^{-i} p_i  C_{i}=\sum\limits_{i=0}^{\log_3 L-1} 3^{-i} p_i  C_{i}.
\end{equation}
Considering the actual duration of data transmission after pilot-based channel estimation, the per-cell net sum-rate for a given normalized coherence time $N_{coh}$ can be expressed as 
\begin{equation*}
C_{net}(\mathbf{p},N_{coh})=\frac{N_{coh}-N_{pil}(\mathbf{p})}{N_{coh}}C_{sum}(\mathbf{p}).
\end{equation*}

\subsection{Optimal Pilot Assignment for Multi-User Multi-Cell System}\label{subsection:Optimal Pilot Assignment for Multi-User Multi-Cell System}

There are two main findings in our previous work \cite{sohn2015pilots}. 

\textit{First Finding:} With fixed pilot length $N_{p0}$ or, equivalently, fixed time duration allocated to pilot-based channel estimation, the closed form solution for the optimal pilot assignment vector maximizing the per-cell sum rate is found. The optimal solution is formulated as
\begin{equation}\label{Sohn2015:Thm1}
\mathbf{p}_{opt,K}'(N_{p0}) = \underset{\mathbf{p}\in \Omega(N_{p0},K)}{\arg\max}\ C_{sum}(\mathbf{p}),
\end{equation}
where $\Omega(N_{p0},K)=\{\mathbf{p}\in P_{L,K} \: : \: N_{pil}(\mathbf{p})=N_{p0}\}$
is the set of all valid pilot assignment vectors $\mathbf{p}$, with length of $N_{pil}(\mathbf{p})=N_{p0}$.
The solution is given by
$\mathbf{p}'_{opt,K}(N_{p0})=(p'_{0},\cdots,p'_{\log_3 L -1})$ where
\begin{equation} \label{Thm3 result}
p'_{i} =
\begin{cases}
\displaystyle\sum_{t=0}^{i}K3^{t}-\frac{N_{p0}-K}{2} & i=\chi(N_{p0},K) \\
3\left(\displaystyle\frac{N_{p0}-K}{2}-\displaystyle\sum_{t=0}^{i-2}K3^{t}  \right) & i=\chi(N_{p0},K)+1 \\
0 & \textrm{otherwise}
\end{cases}
\end{equation}
with $\chi(N_{p0},K)=\min\{k \: : \: \sum_{i=0}^{k}K3^{i}>\frac{N_{p0}-K}{2} \}$
being the first non-zero position of $\mathbf{p}'_{opt,K}(N_{p0})$, i.e., the depth of the least deep leaf node.

\textit{Second Finding:} The remaining question is that for a given channel coherence time, how much time duration should be allocated to pilot-based channel estimation. Equivalently, with a given $N_{coh}$, what is the optimal duration $N_{p0}$ for the pilot transmission which maximizes the per-cell net sum-rate. This is given in \cite{sohn2015pilots}.
For some practical scenarios, the optimal solution has been shown to be far different from the conventional full pilot reuse. Considerable net sum-rate gains were observed in the simulation results.



\section{Pilot Assignment Strategy for net-WSR maximization}\label{Section:OptimizationProblem}

In this section, we provide analysis on pilot assignment strategy for net-WSR maximization. 
The scenario of using orthogonal pilot sequences possibly larger than $K$ is considered, while users within the same cell are guaranteed to have orthogonal pilots. For ease of analysis, users in different priority groups are assumed to use orthogonal pilots to each other.

\subsection{User Prioritizing}

We assume $K$ users in each cell are divided into 2 groups depending on the priority (Section \ref{Section:Further Comments} deals with the general case where the number of priority groups is greater than 2).
Let $\alpha$ be the ratio of the number of $1^{st}$ priority users to the total number of users ($0 < \alpha < 1$). Let $K_1$ be the number of $1^{st}$ priority users in each cell, i.e., $K_1 = \alpha K$. Similarly, $K_2$ is the number of $2^{nd}$ priority users in each cell, i.e., $K_2 = (1-\alpha) K$.
Considering the scenario with different weights on different user groups, let $\omega$ be the weight on the $1^{st}$ priority group ($0.5 \leq \omega < 1$).

\subsection{Pilot Assignment Vector}

For mathematical analysis on pilot assignment strategy, we utilize tools established in \cite{sohn2015pilots}: 3-way partitioning and pilot assignment vector representation. 
First, since we assume users in different priority groups have orthogonal pilots, pilot assignment for the system can be divided into 2 independent sub-assignments: pilot assignment for $1^{st}$ priority group and then for $2^{nd}$ priority group. Various available pilot assignments for each group can be easily expressed in a vector form by using the definition in Section III. 
$P_{L,K_1}$ represents the set of valid pilot assignments for $1^{st}$ priority group while $P_{L,K_2}$ is used for $2^{nd}$ priority group.

The set of valid pilot assignments for the entire network ($LK$ users composed of $LK_1$ users in $1^{st}$ priority and $LK_2$ users in $2^{nd}$ priority) needs to be defined.
The set of valid pilot assignments for $L$ cells where each cell has $K_{1}$ users in $1^{st}$ priority group and $K_{2}$ users in $2^{nd}$ priority group can be defined as
\begin{eqnarray*} 
\widetilde P_{L,K, \alpha} = \{\bm{[}\mathbf{p}_{1}, \mathbf{p}_{2}\bm{]} \:\: : \enspace 
\mathbf{p}_{1} \in P_{L,K_{1}}, \mathbf{p}_{2} \in P_{L,K_{2}} \}.
\end{eqnarray*}
Therefore, in order to consider all possible pilot assignments $\mathbf{p} \in \widetilde P_{L,K, \alpha}$ for our system, we need to check possible pairs of $[\mathbf{p}_{1}, \mathbf{p}_{2}]$ where $\mathbf{p}_{1}$ is a valid assignment for $1^{st}$ priority users and $\mathbf{p}_{2}$ is a valid assignment for $2^{nd}$ priority users.

\subsection{Problem Formulation}

For $\mathbf{p}_{1} \in P_{L,K_{1}}$ and $\mathbf{p}_{2} \in P_{L,K_{2}}$, the corresponding per-cell WSR value is expressed as
\begin{equation} \label{weighted sum rate}
C_{wsr}(\mathbf{p}_{1},\mathbf{p}_{2})= \omega C_{sum}(\mathbf{p}_{1}) + (1-\omega) C_{sum}(\mathbf{p}_{2}) .
\end{equation}
Considering the fact that data transmission is available for the portion of coherence time not allocated to pilot training, the per-cell net-WSR is written as
\begin{align*} 
C_{net,wsr}&(\mathbf{p}_{1},\mathbf{p}_{2}, N_{coh})  \\
&=\frac{N_{coh}- [N_{pil}(\mathbf{p}_1)+N_{pil}(\mathbf{p}_2)]}{N_{coh}} C_{wsr}(\mathbf{p}_{1},\mathbf{p}_{2}).
\end{align*}
Now, we formulate our optimization problem. Let $L$, $K$, $\alpha$, and $\omega$ be fixed.
For a  given $N_{coh}$ value, we want to find out optimal $\mathbf{p}_{1} \in P_{L,K_{1}}$ and $\mathbf{p}_{2} \in P_{L,K_{2}}$ which maximize $C_{net,wsr}(\mathbf{p}_{1},\mathbf{p}_{2}, N_{coh})$. In other words, the optimal pilot assignment vector $\mathbf{p}_{opt}(N_{coh}) 
$ which maximizes net-WSR is
\begin{align*}
\mathbf{p}_{opt}(N_{coh}) &= [\mathbf{p}_{opt}^{(1)}(N_{coh}), \mathbf{p}_{opt}^{(2)}(N_{coh})]\nonumber\\
&\triangleq \underset{[\mathbf{p}_{1},\mathbf{p}_{2}]\in \widetilde P_{L,K,\alpha}}{\arg\max}\ C_{net,wsr}(\mathbf{p}_{1},\mathbf{p}_{2},N_{coh}).
\end{align*}
For a given $N_{coh}$, the $\mathbf{p}_{opt}$ function outputs the optimal pilot assignment pair $[\mathbf{p}_{1},\mathbf{p}_{2}]\in \widetilde P_{L,K,\alpha}$, where $\mathbf{p}_{opt}^{(i)}(N_{coh})$ is the optimal assignment for the $i^{th}$ priority group.

\section{Closed-form solution for optimal pilot assignment}\label{Section:OptimizationSolution}

In this section, the closed-form solution to the net-WSR maximization problem is presented. 
The solution can be obtained in two steps, which are dealt with in the following two subsections, respectively.
Subsection A establishes the optimal pilot assignment rule which maximizes the WSR when the total available number of pilots is given.
Subsection B finds the optimal total pilot length which maximizes the net-WSR for a given $N_{coh}$. Combining these two solutions, we obtain $\mathbf{p}_{opt}(N_{coh})$.

\subsection{Optimal Pilot Assignment Vector under a Total Pilot Length Constraint}

Under a constraint on the  total pilot length $N_{pil}(\mathbf{p}_{1}) + N_{pil}(\mathbf{p}_{2}) = T$, let us find the pilot assignment vector which maximizes $C_{wsr}(\mathbf{p}_{1}, \mathbf{p}_{2})$. 
This sub-problem can be formulated as
\begin{align*}
\widetilde P_{opt}(&T) = \Big\{\bm{[}\mathbf{p}_{1},\mathbf{p}_{2}\bm{]} \in \Theta(T) \: : \  \\ &C_{wsr}(\mathbf{p}_{1},\mathbf{p}_{2}) \geq  C_{wsr}(\mathbf{p'}_{1},\mathbf{p'}_{2}) 
 \:\:\:\: \forall \bm{[}\mathbf{p'}_{1},\mathbf{p'}_{2}\bm{]} \in \Theta(T) \Big\}
\end{align*}
where 

$\Theta(T)= \{[\mathbf{p}_{1}, \mathbf{p}_{2}] \in \widetilde P_{L,K, \alpha} \:  :  \:
 N_{pil}(\mathbf{p}_{1})+N_{pil}(\mathbf{p}_{2})=T \}$.
This optimization problem might have multiple solutions. We thus define the set $\widetilde P_{opt}(T) $ of optimal pilot assigning strategies. 
Lemmas \ref{Lemma:Lemma1} and \ref{Lemma:Lemma2} given below specify $\widetilde P_{opt}(T)$, the set of optimal pilot assignment vectors under the pilot length constraint. Before stating our main Lemmas, we introduce some short-hand notations:
\begin{align*}
& \mathbb{Z} =\{\cdots, -2, -1, 0, 1, 2, \cdots \},\nonumber\\
& B(T) = \max(K_1, T-LK_2/3), \nonumber\\
& F(T) = \min(T-K_2, LK_1/3), \nonumber\\
& S_0(T) = \{B(T), B(T) + 2, \cdots, F(T)\}, \nonumber\\
& S_1(T) = S_0(T) \setminus \{F(T)\},\nonumber\\
& g_T(t) = 3^{\chi(t, K_1)-\chi(T-t-2, K_2)} \frac{1-\omega}{\omega}.
\end{align*}
where 
$\chi(N_{p0},K)=\min\{k \: : \: \sum_{i=0}^{k}K3^{i}>\frac{N_{p0}-K}{2} \}$ as defined in section III-C.

The total pilot length $N_{pil}(\mathbf{p}_{1}) + N_{pil}(\mathbf{p}_{2}) = T$ can be decomposed into two parts: pilot length $N_{pil}(\mathbf{p}_{1}) = t$ for $1^{st}$ priority group and pilot length $N_{pil}(\mathbf{p}_{2}) = T-t$ for $2^{nd}$ priority group. 
From \cite{sohn2015pilots}, $N_{pil}(\mathbf{p}_{1}) \in \{K_1, K_1 + 2, \cdots, LK_1 / 3 \}$ and $N_{pil}(\mathbf{p}_{2}) \in \{K_2, K_2 + 2, \cdots, LK_2 / 3 \}$ holds for $[\mathbf{p}_{1}, \mathbf{p}_{2}] \in \widetilde P_{L,K, \alpha}$, so that we can obtain the set of possible [$N_{pil}(\mathbf{p}_{1})$, $N_{pil}(\mathbf{p}_{2})$] pairs illustrated in Table \ref{Table:Pilot length pairs}. As seen in the Table, $B(T)$ represents the minimum possible value assigned for $1^{st}$ priority group, and $F(T)$ is the maximum possible value assigned for $1^{st}$ priority group (Note that $B(T) \leq F(T)$ by definition). $S_0(T)$ represents the set of possible values assigned for $1^{st}$ priority group.
Later, it can be seen that $g_T : S_1(T)\rightarrow \mathbb{R}$ is defined for comparing $C_{wsr}$ values of different assignments.
Now we state our main Lemmas.

\begin{table}
	\small
\caption{Possible pilot length pairs for each priority group}
\centering
\label{Table:Pilot length pairs}
	\setlength\tabcolsep{1.5pt} 
\begin{tabular}{|c|c|}
\hline
$N_{pil}(\mathbf{p}_{1}) = t$ & $N_{pil}(\mathbf{p}_{2}) = T-t$ \tabularnewline
\hline
$B(T)$ & $T-B(T)$ \tabularnewline
$B(T)+2$ & $T-B(T)-2$ \tabularnewline
$\vdots$ & $\vdots$ \tabularnewline
$F(T) - 2$ & $T-F(T)+2$ \tabularnewline
$F(T)$ & $T-F(T)$ \tabularnewline
\hline
\end{tabular}
\end{table}

\begin{lemma}\label{Lemma:Lemma1}
If $\log_3 \frac{\omega}{1-\omega} \notin \mathbb{Z}$, then
the set $\widetilde P_{opt}(T)$ of optimal pilot assignment vectors maximizing $C_{wsr}$ is
\begin{align}\label{Lemma1 result}
\widetilde P_{opt}(T)&=\Big\{\bm{[}\: \mathbf{p}'_{opt, K_{1}}(\rho(T)),\mathbf{p}'_{opt, K_{2}}(T-\rho(T)) \:\bm{]}\Big\}
\end{align}
where
\begin{equation}\label{def:rho_T}
\rho(T) =
\begin{cases}
B(T), \ \ \ \ \ \ \  \ \text{if  } S_1(T) = \emptyset \text{ or } g_T(B(T)) > 1 \\
F(T), \ \ \ \ \ \ \ \ \ \ \ \ \ \ \text{ else if }  g_T(F(T) - 2) < 1\\
min \{ t\in S_1(T)  \:\: : \enspace g_T(t) \geq 1 \}, \ \ \ \ \text{otherwise.}
\end{cases}
\end{equation}


\end{lemma}

As stated in Lemma \ref{Lemma:Lemma1}, there exists unique optimal vector $\bm{[}\: \mathbf{p}'_{opt, K_{1}}(\rho(T)),\mathbf{p}'_{opt, K_{2}}(T-\rho(T)) \:\bm{]}$ for every $\omega$ satisfying $\log_3 \frac{\omega}{1-\omega} \notin \mathbb{Z}$. However, in the case of $\log_3 \frac{\omega}{1-\omega} \in \mathbb{Z}$, we have possibly multiple optimal solutions as specified in the following Lemma.

\begin{lemma} \label{Lemma:Lemma2}
If $\log_3 \frac{\omega}{1-\omega} \in \mathbb{Z}$, then
the set $\widetilde P_{opt}(T)$ of optimal pilot assignment vectors maximizing $C_{wsr}$ is
\begin{align}\label{Lemma2 result}
\widetilde P_{opt}(T)= \Big\{
\bm{[}\: \mathbf{p}'_{opt, K_{1}}(t),\mathbf{p}'_{opt, K_{2}}(T-t) \:\bm{]}  \:\: : & \nonumber\\ \: t \in \{\rho(T),  \rho(T)+2,  \cdots, & \mu(T)\} \Big\}
\end{align} 
where $\rho(T)$ is as defined in (\ref{def:rho_T}) and
\begin{equation*}
\mu(T) =
\begin{cases}
B(T), \ \ \ \ \ \ \ \ \ \ \ \text{if  } S_1(T) = \emptyset \text{ or } g_T(B(T)) > 1 \\
F(T), \ \ \ \ \ \ \ \ \ \ \ \ \ \ \ \ \ \  \text{else if } g_T(F(T) - 2) < 1\\
min \{ t\in S_1(T) : g_T(t) \leq 1 \} + 2, \ \ \ \ \text{otherwise.}
\end{cases}
\end{equation*}
\end{lemma}

We defer the proofs of Lemmas 1 and 2 to Appendix \ref{Appendix:Proofs of Lemmas 1 and 2}. 
These two lemmas lead to our first main theorem, which 
specifies an element of $\widetilde P_{opt}(T)$ as a function of $T$.

\begin{theorem} \label{Theorem:Theorem1}
For given total pilot length $T \in \{K, K+2, \cdots, LK/3\}$, using pilot assignment vector $\mathbf{p}'_{opt, K_{1}}(\rho(T))$ for $1^{st}$ priority group and $\mathbf{p}'_{opt, K_{2}}(T-\rho(T))$ for $2^{nd}$ priority group maximizes the WSR. This optimal solution allocates $\rho(T)$ pilots to $1^{st}$ priority group and $T-\rho(T)$ to $2^{nd}$ priority group. In other words, 
\begin{equation}\label{Theorem:Theorem1 result}
\bm{[}\: \mathbf{p}'_{opt, K_{1}}(\rho(T)),\mathbf{p}'_{opt, K_{2}}(T-\rho(T)) \:\bm{]} \in \widetilde P_{opt}(T).
\end{equation}
\end{theorem}

\begin{proof}[Proof\nopunct] 
From (\ref{Lemma1 result}) and (\ref{Lemma2 result}), we can confirm that (\ref{Theorem:Theorem1 result}) holds, irrespective of $\omega$ value.
\end{proof}


When $\log_3 \frac{\omega}{1-\omega} \notin \mathbb{Z}$ holds, $\widetilde P_{opt}(T)$ contains unique element 
as stated in (\ref{Lemma1 result}) or (\ref{Theorem:Theorem1 result}). In the case of $\log_3 \frac{\omega}{1-\omega} \in \mathbb{Z}$, $\widetilde P_{opt}(T)$ might have multiple elements as in (\ref{Lemma2 result}), but we can guarantee the existence of an element stated in (\ref{Theorem:Theorem1 result}).
Here, we attempt to get some insight on $\rho(T)$ in (\ref{def:rho_T}), the optimal pilot length for $1^{st}$ priority group. The following proposition suggests an alternative expression for $\rho(T)$, depending on the range of $T$.

\begin{prop}\label{Prop:rho_T expression}
Let $s = \ceil[\big]{\log_3 \frac{w}{1-w}}$.
If $3^s \geq L/3$,  
\begin{equation}\label{def:rho_T_alternative}
\rho(T) =
\begin{cases}
T-K_2 & \text{if  } K \leq T \leq K_2 +  \frac{LK_1}{3} \\
LK_1/3 & \text{if  } K_2 + \frac{LK_1}{3} < T \leq \frac{LK}{3}
\end{cases}
\end{equation}
Otherwise (i.e., $3^s < L/3$), 
\begin{equation}\label{def:rho_T_alternative2}
\rho(T) =
\begin{cases}
T-K_2 & \text{if  } K \leq T \leq K_2 +  3^s K_1 \\
\phi(T) & \text{if  } K_2 + 3^s K_1 < T < \frac{LK_1}{3} + \frac{LK_2}{3^{s+1}} \\
LK_1 / 3 & \text{if  } \frac{LK_1}{3} + \frac{LK_2}{3^{s+1}} \leq T \leq \frac{LK}{3} \\
\end{cases}
\end{equation}
where
\begin{align*}\label{def:phi_T}
\phi(T) &=
\begin{cases}
3^{V(T)+s-1} K_1 & \text{if  } T \leq 3^{V(T)+s-1}K_1 +  3^{V(T)} K_2 \\
T - 3^{V(T)} K_2 & \text{if  } T > 3^{V(T)+s-1}K_1 +  3^{V(T)} K_2, \\
\end{cases}\nonumber\\
V(T) &=\min\{i \in \{0, 1, \cdots, log_3 L -1-s \} \: : \\ 
& \ \ \ \ \ \ \ \ \ \ \ \ \ \ \ \ \ \ \ \ \ \ \ \ \ \ \ \ \ \ \ \ \ \ \ \ \ \: T \leq 3^{s+i}K_1 + 3^i K_2\}.
\end{align*}
\end{prop}

\begin{table}[htbp]
	\small
	\centering
	\caption{Optimal pilot lengths for each priority group}
	\subfloat[$3^s \geq L/3$ case]{%
		\label{Table: Pattern on optimal assignment_1}
		\begin{tabular}{c|c|c}
			\hline
			$\boldsymbol{T}$ & $\boldsymbol{\rho}(\boldsymbol{T})$ & $\boldsymbol{T} - \boldsymbol{\rho}(\boldsymbol{T})$ \\

			\hline
			$K$ & $ K_1 $ & $$ \\
			$K+2$ & $K_1 + 2$ & $$ \\
			$\vdots$ & $\vdots$ & $K_2$ \\
			$K_2 + \frac{LK_1}{3}$ & $\frac{LK_1}{3}$ & $$  \\
			\hline
			$K_2 + \frac{LK_1}{3} + 2$ & $$ & $K_2 + 2$ \\
			$\vdots$ & $\frac{LK_1}{3}$ & $\vdots$ \\
			$\frac{LK}{3}$ & $$ & $\frac{LK_2}{3}$ \\
			\hline
		\end{tabular}%
	}\hspace{1cm}
	\subfloat[$3^s < L/3$ case]{%
\label{Table: Pattern on optimal assignment_2}
		\begin{tabular}{c|c|c|c|c}
			\hline
			$\boldsymbol{T}$ & $\boldsymbol{\rho}(\boldsymbol{T})$ & $\boldsymbol{T} - \boldsymbol{\rho}(\boldsymbol{T})$ & $\boldsymbol{d_1}$ & $\boldsymbol{d_2}$ \\
			\hline
			$K$ & $K_1$ & $$ & 0 &  \\
			$K+2$ & $K_1 + 2$ & $$ & 0 & \\
			$\vdots$ & $\vdots$ & $K_2$ & $\vdots$ & 0 \\
			$K_2 + 3^{s}K_1 - 2$ & $3^{s}K_1 - 2$ & $$ & $s-1$ &  \\
			$K_2 + 3^{s}K_1$ & $3^{s}K_1$ & $$ & $s$ &  \\
			\hline
			$K_2 + 3^{s}K_1 + 2$ & $$ & $K_2 + 2$ &  & 0 \\
			$\vdots$ & $3^{s}K_1$ & $\vdots$ & $s$ & $\vdots$\\
			$3 K_2 + 3^{s}K_1 - 2$ & $$ & $3 K_2 - 2 $ &  & $0$ \\
			$3 K_2 + 3^{s}K_1$ & $$ & $3 K_2 $ &  & $1$ \\
			\hline
			$3 K_2 + 3^{s}K_1 + 2$ & $3^{s}K_1 + 2$ & $$ & $s$ &  \\
			$\vdots$ & $\vdots$ & $3 K_2 $   & $\vdots$ & $1$  \\
			$3 K_2 + 3^{s+1}K_1$ & $3^{s+1}K_1$ & $$ & $s+1$ &  \\
			\hline
			$\vdots$ & $\vdots$ & $\vdots$  & $\vdots$ & $\vdots$\\
			\hline
			$\frac{LK_1}{3} + \frac{LK_2}{3^{s+1}}$ & $$ & $\frac{LK_2}{3^{s+1}}$ & & \\
			$\vdots$ & $\frac{LK_1}{3}$ & $\vdots$ & & \\
			$\frac{LK}{3}$ & $$ & $\frac{LK_2}{3}$ & & \\
			\hline
		\end{tabular}%
	}
\end{table}

The detailed proof for Proposition \ref{Prop:rho_T expression} is in Appendix \ref{Appendix:Proofs of Propositions 1 and 2}, but  here we provide a brief explanation on the result that is consistent with our intuition.
As $T$ increases, variation of $\rho(T)$ has a pattern illustrated in Tables \ref{Table: Pattern on optimal assignment_1}. 
When $\omega$ is sufficiently large ($3^{s} \geq L/3$ case, Table \ref{Table: Pattern on optimal assignment_1}), the WSR is mostly determined by the sum rate of the $1^{st}$ priority group. Therefore, the WSR-maximizing solution first allocates additional pilot resources to 
         the $1^{st}$ priority group up to its maximum pilot length $LK_1 / 3$. After the $1^{st}$ group gets maximum available pilot resources, additional pilot is dedicated to the $2^{nd}$ priority group.
In the case of $3^{s} < L/3$ (Table \ref{Table: Pattern on optimal assignment_2}), weight $\omega$ on the $1^{st}$ priority group is not large enough so that the optimal pilot resource allocation rule has an alternative pattern for two priority groups, as illustrated below.

\begin{figure}
	\centering
   \includegraphics[width=80mm]{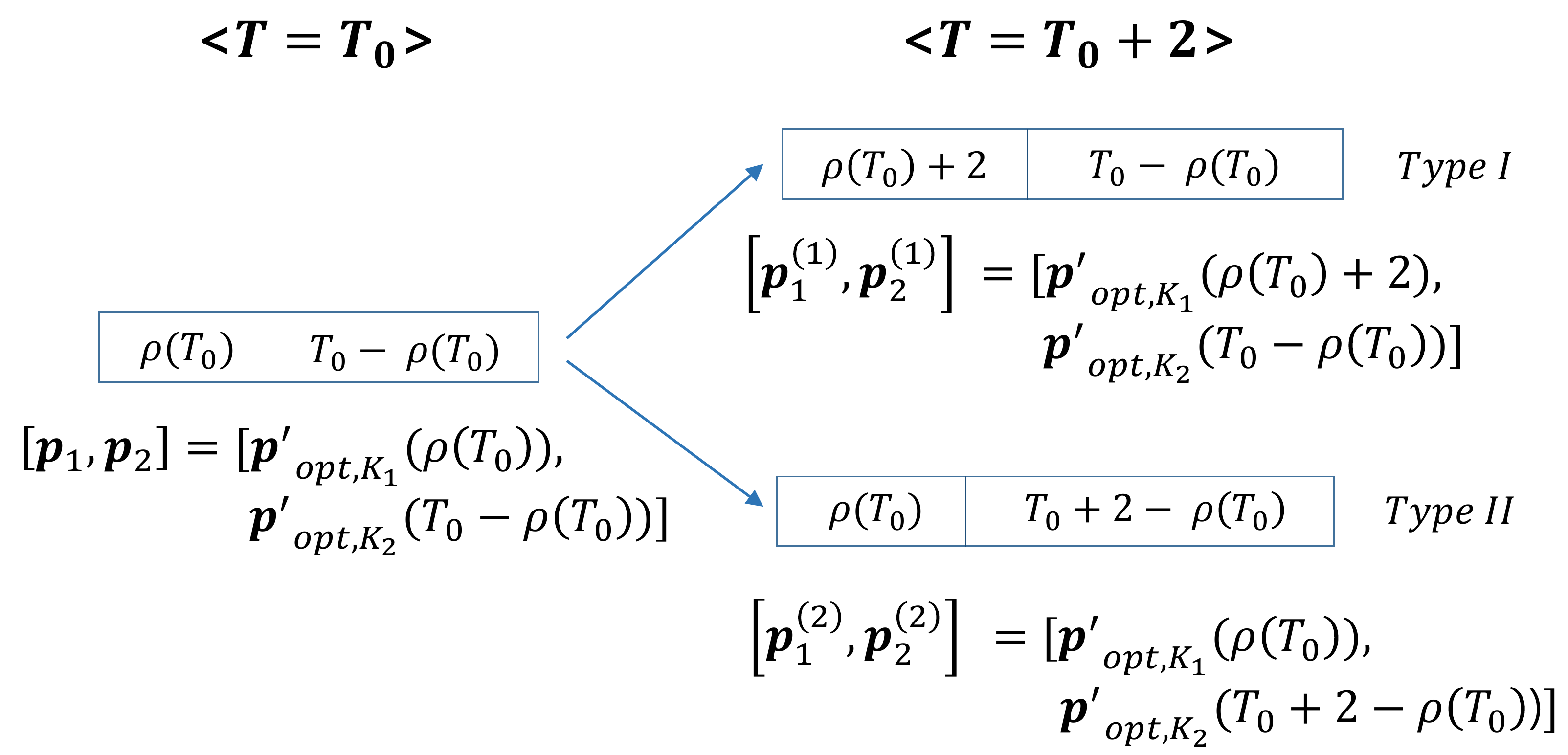}
    \caption{Analysis on Proposition 1}
    \label{Fig:proposition1_analysis}
\end{figure}


Consider the optimal assignment for $T=T_0$, as shown in left side of Fig. \ref{Fig:proposition1_analysis}. 
When additional pilot resource is available, we can consider two types of assignments for $T=T_0 + 2$: Type I allocates extra pilots to $1^{st}$ priority group, while Type II allocates extra pilots to $2^{nd}$ priority group. 
From Corollary 1 of \cite{sohn2015pilots}, $\mathbf{p}_{1}^{(1)}$ can be obtained by tossing 1 from the left-most non-zero element of $\mathbf{p}_{1}$ to increase the adjacent element by 3 (Similar relationship exists for $\mathbf{p}_{2}^{(2)}$ and $\mathbf{p}_{2})$. Denote the left-most non-zero element of $\mathbf{p}_{1}$ and $\mathbf{p}_{2}$ by $d_1 = \chi(\rho(T_0), K_1)$ and $d_2 = \chi(T-\rho(T), K_2)$, respectively. Then from the WSR expression (\ref{weighted sum rate}), Type I allocation increases the WSR by $\omega L 3^{-d_1} (C_{d_1+1}-C_{d_1})$ and Type II allocation yields an increase by $(1-\omega) L 3^{-d_2} (C_{d_2+1}-C_{d_2})$, compared to the WSR of $[\mathbf{p}_{1}, \mathbf{p}_{2}]$ allocation.



 
Therefore, when the total pilot length is $T_{0}+2$, the optimal solution is chosen by comparing $\omega L 3^{-d_1} (C_{d_1+1}-C_{d_1})$ and $(1-\omega) L 3^{-d_2} (C_{d_2+1}-C_{d_2})$. Using $C_{i+1} - C_{i} \simeq const.$, this  reduces to comparing $s$ and $d_1 - d_2$. In summary, as $T$ increases by $2$, additional $2$ pilots are assigned to either $1^{st}$ or $2^{nd}$ priority group, where the decision is based on the sign of $s - (d_1 - d_2)$. 
For example, consider $T = 3K_2 + 3^s K_1$ in Table \ref{Table: Pattern on optimal assignment_2}. Since $d_1 = s$ and $d_2 = 1$, we have $s > d_1 - d_2$. Therefore, assigning the additional pilot to $1^{st}$ priority group 
is the WSR-maximizing choice, so that $\rho(T+2) = \rho(T) + 2 = 3^s K_1 + 2$ as in the Table.
Note that $d_1$ increases as $\rho(T)$ increases, while $d_2$ increases as $T-\rho(T)$ increases. Therefore, considering the sign of $s - (d_1 - d_2)$, balanced resource allocation occurs as additional pilot resource is allowed.
This alternative allocation occurs until $T < \frac{LK_1}{3} + \frac{LK_2}{3^{s+1}}$.

When $T = \frac{LK_1}{3} + \frac{LK_2}{3^{s+1}}$, the optimal number of pilots for the $1^{st}$ priority group reaches $\frac{LK_1}{3}$, the maximum possible value. Therefore, similar to the $3^s \geq L/3$ case, for $T$ values greater than the threshold value, the additional pilot is dedicated to the $2^{nd}$ priority group.
From this observation, we can directly obtain the following proposition, which relates the consecutive $\rho(T)$ values as $T$ increases. The proof of the proposition is in Appendix \ref{Appendix:Proofs of Propositions 1 and 2}.

\begin{prop}\label{Prop:rho_T relation}
Either
$\rho(T+2)=\rho(T)$ or $\rho(T+2)=\rho(T) + 2$ holds for $T \in \{K, K+2, \cdots, \frac{LK}{3} - 2 \}$.
\end{prop}

\subsection{Optimal Total Pilot Length for Given Channel Coherence Time}

In this subsection, we solve the second sub-problem: for a given normalized coherence time $N_{coh}$, find the optimal total pilot length $T$ which maximizes the net-WSR. Here, we define the maximum WSR value under the constraint on the total pilot length $T$ as
\begin{equation*}
\bar{C}_{wsr}(T)= \underset{[\mathbf{p}_{1},\mathbf{p}_{2}]\in \Theta(T)}{\max}\ C_{wsr}(\mathbf{p}_{1},\mathbf{p}_{2}).
\end{equation*} 
By using (\ref{Theorem:Theorem1 result}), we can relate $\bar{C}_{wsr}(T)$ and $\bar{C}_{wsr}(T+2)$ as in the following Corollary, proof of which is given in Appendix \ref{Appendix:Proof of Corollary 1}.



\begin{corollary} \label{Corollary:Corollary1}
For $T \in \{K, K+2, \cdots, \frac{LK}{3} - 2 \}$, 
\begin{align*}
\bar{C}_{wsr}(T+2)
=  \bar{C}_{wsr}(T) + \delta_T
\end{align*}
holds where
\begin{align*}
\delta_T &=
\begin{cases}
\delta_T^{(2)} & T \geq \frac{LK_1}{3} + \max\{\frac{L}{3^{s+1}}, 1\} \: K_2 \\
max \{\delta_T^{(1)}, \delta_T^{(2)}\} & \textrm{otherwise,} \nonumber
\end{cases}\\
\delta_T^{(1)}&= \omega L 3^{-d_1} (C_{d_1 + 1} - C_{d_1}),\nonumber\\
\ \ \ \delta_T^{(2)}&= (1-\omega) L 3^{-d_2} (C_{d_2 + 1} - C_{d_2}),\nonumber\\
d_1 &= \chi(\rho(T), K_1), \text{and } d_2 = \chi(T-\rho(T), K_2).
\end{align*}
\end{corollary}
Note that $\{C_{i}\}$ is defined in Section III-A. Corollary \ref{Corollary:Corollary1} implies that $\bar{C}_{wsr}(T)$ is an increasing function of $T$. Now, using Corollary \ref{Corollary:Corollary1}, we proceed to find the solution for the optimal total pilot length for a given $N_{coh}$, as stated in Theorem \ref{Theorem:Theorem2} below. Define 
\begin{align*}
h_{T}(&N_{coh})\\
&\triangleq C_{net,wsr}(\mathbf{p}'_{opt, K_{1}}(\rho(T)),\mathbf{p}'_{opt, K_{2}}(T-\rho(T)), N_{coh}) \\ 
&= \frac{N_{coh}- T}{N_{coh}} \bar{C}_{wsr}(T),
\end{align*} 
which is the maximum net-WSR value for a given pilot length $T$. Note that $h_{T}(N_{coh})$ is an increasing function of $N_{coh}$, which is positive for $N_{coh} > T$ and saturates to $\bar{C}_{wsr}(T)$ as $N_{coh}$ goes to infinity.
Consider a plot of this function for $T = K, K+2, \cdots, LK/3$. Using the fact that $\bar{C}_{wsr}(T)$ is an increasing function of $T$
, we can check that
the curve $h_{T}(N_{coh})$ is above other curves for $K\Delta_{\frac{T-K}{2}} \leq N_{coh} < K\Delta_{\frac{T-K}{2}+1}$ where
\begin{equation*} 
\Delta_n = 
\begin{cases}
0, & n=0\\
\Big\{ 2n+K + \frac{2\bar{C}_{wsr}(2n+K+2)}{\delta_{2n+K}} \Big\} \Big/ K, & 1 \leq n \leq N_{L} \\
\infty, & n=N_{L}+1.
\end{cases}
\end{equation*}
and $N_{L}=\frac{LK/3 - K}{2}$ (detailed analysis is in Appendix \ref{Appendix:Proof of Theorem 2}).
Therefore, $K\Delta_n$ represents consecutive $N_{coh}$ values where optimal assignment changes. Using this definition, we now state our second main theorem, specifying $\mathbf{p}_{opt}(N_{coh})$. The proof of the theorem is in Appendix \ref{Appendix:Proof of Theorem 2}.


\begin{theorem} \label{Theorem:Theorem2}
If the given normalized coherence time $N_{coh}$ satisfies $\Delta_{n} \leq N_{coh}/K < \Delta_{n+1}$, the optimal pilot assignment $\mathbf{p}_{opt}(N_{coh})=[\mathbf{p}_{opt}^{(1)}(N_{coh}), \mathbf{p}_{opt}^{(2)}(N_{coh})]$ that maximizes net-WSR $C_{net,wsr}$ has the following form:
\begin{align}\label{Theorem:Theorem2 result}
\mathbf{p}_{opt}^{(1)}(N_{coh}) &= \mathbf{p}'_{opt, K_1}(\rho(2n+K))\nonumber\\
\mathbf{p}_{opt}^{(2)}(N_{coh}) &= \mathbf{p}'_{opt, K_2}(2n+K - \rho(2n + K)).
\end{align}
Also, the optimal number of pilots is 
\begin{equation*}
N_{pil}(\mathbf{p}_{opt}^{(1)}(N_{coh})) + N_{pil}(\mathbf{p}_{opt}^{(2)}(N_{coh})) = 2n + K.
\end{equation*}
\end{theorem}

For any positive $N_{coh}$ value, there exists an integer $\\0 \leq n \leq N_{L}$ such that 
$N_{coh} \in [K\Delta_n , K\Delta_{n+1})$. 
The optimal assignment for a given $N_{coh}$ is specified by the corresponding $n$ values, as in (\ref{Theorem:Theorem2 result}). In other words, the optimal pilot assignment for the $1^{st}$ priority group turns out to be $\mathbf{p}'_{opt, K_1}(\rho(2n+K))$, while for the $2^{nd}$ priority group we have $\mathbf{p}'_{opt, K_2}(2n+K - \rho(2n + K))$. Combining with (\ref{Thm3 result}), we can obtain the exact components of the optimal pilot vector.
Also, we can figure that the optimal number of pilots utilized in the system is $N_{pil}(\mathbf{p}_{opt}^{(1)}(N_{coh}))+N_{pil}(\mathbf{p}_{opt}^{(2)}(N_{coh}))=2n+K$ for $N_{coh}$ values satisfying $\Delta_{n} \leq N_{coh}/K < \Delta_{n+1}$.




\section{Simulation Results}\label{Section:Simulation}


In this section, the effect of utilizing optimal assignment is analyzed, based on simulation results. The optimal solution depends on $C_{i}$ values in (\ref{sum rate}), which need to be obtained by simulation. $C_{i}$ depends on $\beta$ terms in (\ref{achievableR}), which is a function of distance between interfering users. Since the location of each user within a cell is assumed to be uniform-random, the $\beta$ terms need to be generated in pseudo-random manner by simulation. 
Following the settings in \cite{bjornson2016massive}, the signal decay exponent is $\gamma = 3.7$, the cell radius is $r$ meters, and the cell-hole radius is $0.14r$.
Note that $\{C_i\}$ values do not depend on $r$. The computation of $C_{i}$ values proceeds by taking average of 100,000 pseudo-random trials on user location.  
The system with $L=81$ cells and $K=10$ users in each cell was considered.


Based on the simulation result of $C_{i}$, WSR values for various possible pilot assignments can be computed. From the WSR data, the optimal pilot assignment which maximizes net-WSR $C_{net,wsr}$ for a given $N_{coh}$ can be obtained, as listed in Table \ref{Table:Optimal assignment}. 
$K=10$ and $\alpha=0.2$ imply
that each cell has 2 users with $1^{st}$ priority and 8 users with $2^{nd}$ priority.
We can check that the list of Table \ref{Table:Optimal assignment} is consistent with the mathematical result in Theorem \ref{Theorem:Theorem2}.
For example, when $19 \leq N_{coh} < 23$ (or $\Delta_1 = 1.9 \leq N_{coh}/K < 2.3 = \Delta_2$), according to (\ref{Theorem:Theorem2 result}) with $n=1$, we have $\mathbf{p}_{opt}^{(1)}(N_{coh})=\mathbf{p}'_{opt, K_1}(\rho(K+2))=\mathbf{p}'_{opt, K_1}(4)=(1,3,0,0)$, 
where the last two equalities are from (\ref{def:rho_T}) and (\ref{Thm3 result}), respectively.
Similarly, we obtain 
$\mathbf{p}_{opt}^{(2)}(N_{coh}) = (8,0,0,0)$ which coincide with the result of Table \ref{Table:Optimal assignment}.
In practice, utilizing $N_{pil}(\mathbf{p_{1}}) + N_{pil}(\mathbf{p_{2}}) = 12$ pilots by applying
pilot assignment $(1,3,0,0)$ for $1^{st}$ priority group and $(8,0,0,0)$ for $2^{nd}$ priority group 
maximizes the net-WSR of the system, when the given $N_{coh}$ satisfies $\Delta_1 = 1.9\leq N_{coh}/K < 2.3 = \Delta_2$.

\begin{table}
	\small
\caption{Optimal pilot assignment ($L=81,K=10, \alpha=0.2, \omega=0.7$)}
\centering
\label{Table:Optimal assignment}
\begin{tabular}{c|c|c|c}
\hline
 & $\mathbf{p}_{1}$ & $\mathbf{p}_{2}$ & $N_{pil}(\mathbf{p}_{1}) $\tabularnewline
 $N_{coh}/K$ & $=\mathbf{p}_{opt}^{(1)}(N_{coh})$  & $=\mathbf{p}_{opt}^{(2)}(N_{coh})$  & $+ N_{pil}(\mathbf{p}_{2})$\tabularnewline

\hline
$0\sim 1.9$ & $(2,0,0,0)$ & $(8,0,0,0)$ & 10 \tabularnewline
$1.9\sim 2.3$ & $(1,3,0,0)$ & $(8,0,0,0)$ & 12 \tabularnewline
$2.3\sim 4.3$ & $(0,6,0,0)$ & $(8,0,0,0)$ & 14 \tabularnewline
$4.3\sim 4.7$ & $(0,6,0,0)$ & $(7,3,0,0)$ & 16 \tabularnewline
$4.7\sim 5.1$ & $(0,6,0,0)$ & $(6,6,0,0)$ & 18 \tabularnewline
$\vdots$ & $\vdots$ & $\vdots$ & $\vdots$ \tabularnewline
$122\sim$ & $(0,0,0,54)$ & $(0,0,0,216)$ & 270 \tabularnewline
\hline
\end{tabular}
\end{table}


A certain trend in optimal pilot assignment is revealed in Table \ref{Table:Optimal assignment}. As $N_{coh}/K$ value increases, optimal pilot assignment for one priority group changes while optimal assignment for the other priority group remains the same. For example, optimal pilot assignment for $1^{st}$ priority group changes from $(1,3,0,0)$ to $(0,6,0,0)$ while 
optimal assignment for $2^{nd}$ priority group remains as $(8,0,0,0)$, when $N_{coh}/K$ value changes near $2.3$. 
On the other hand, optimal pilot assignment for $1^{st}$ priority group remains as $(0,6,0,0)$, while optimal assignment for $2^{nd}$ priority group changes from $(7,3,0,0)$ to $(6,6,0,0)$, when $N_{coh}/K$ value changes near $4.7$. 
This result is consistent with the message of Theorem \ref{Theorem:Theorem2}. It is shown that the optimal assignment has the form of (\ref{Theorem:Theorem2 result}), with $n$ increasing as $N_{coh}/K$ grows. Combining with Proposition \ref{Prop:rho_T relation}, we see that the optimal pilot length of the $1^{st}$ priority group either increases by $2$ or remains the same, as $n$ grows. In other words, as $N_{coh}/K$ increases, we are allowed to use two extra pilots, and which group to allocate the extra pilots is based on $\alpha$ and $\omega$.


\begin{figure}[!t]
\centering
        \includegraphics[height=40mm]{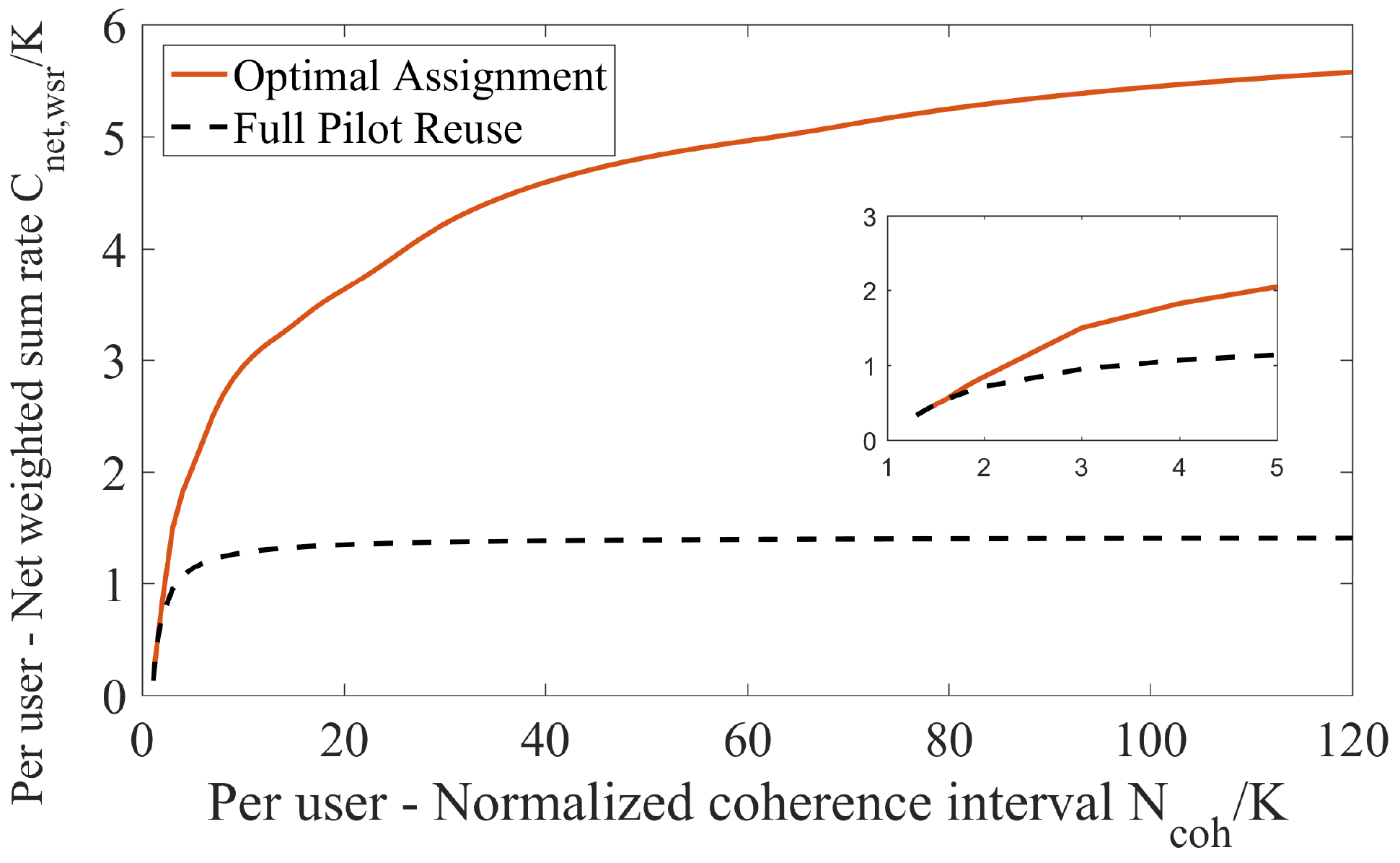}
    \caption{Per-user Net-WSR for optimal/conventional pilot assignments ($L=81,K=10, \alpha=0.2, \omega=0.9$)}
    \label{Fig:alpha=0.2_omega=0.9_netWSR_Simulation}
    
\end{figure}



Now, we compare the net-WSR values of optimal assignment and conventional assignment. Here, conventional assignment means assigning $K$ orthogonal pilots reused in each cell.
Fig. \ref{Fig:alpha=0.2_omega=0.9_netWSR_Simulation} shows net-WSR comparison for $L=81, K=10, \alpha=0.2$, and $\omega=0.9$. As per user - normalized coherence interval $N_{coh}/K$ increases, optimal assignment has substantial net-WSR gains compared to conventional assignment.
In the inset of Fig. \ref{Fig:alpha=0.2_omega=0.9_netWSR_Simulation}, we have a plot focusing on low $N_{coh}/K$ values. 
For $N_{coh}/K$ greater than $1.7$, optimal assignment beats the conventional full reuse. In the case of maximizing per-user net-rate $C_{net}/K$ in \cite{sohn2015pilots}, optimal assignment beats full reuse for $N_{coh}/K$ greater than $6.2$. Therefore, departing from conventional wisdom is beneficial for a wider $N_{coh}/K$ range, compared to the scenario where maximizing the net-rate was the objective.  A considerable increase of per-user net-WSR is seen even at low $N_{coh}/K$ values. For $N_{coh}/K=5, 10$ and $20$, optimal assignment has $79.8\%$, $130.2\%$ and $169.0\%$ higher net-WSRs than conventional assignment.

\begin{figure}
	\centering
	\subfloat[][$\alpha=0.2, \omega=0.8$]{\includegraphics[height=45mm ]{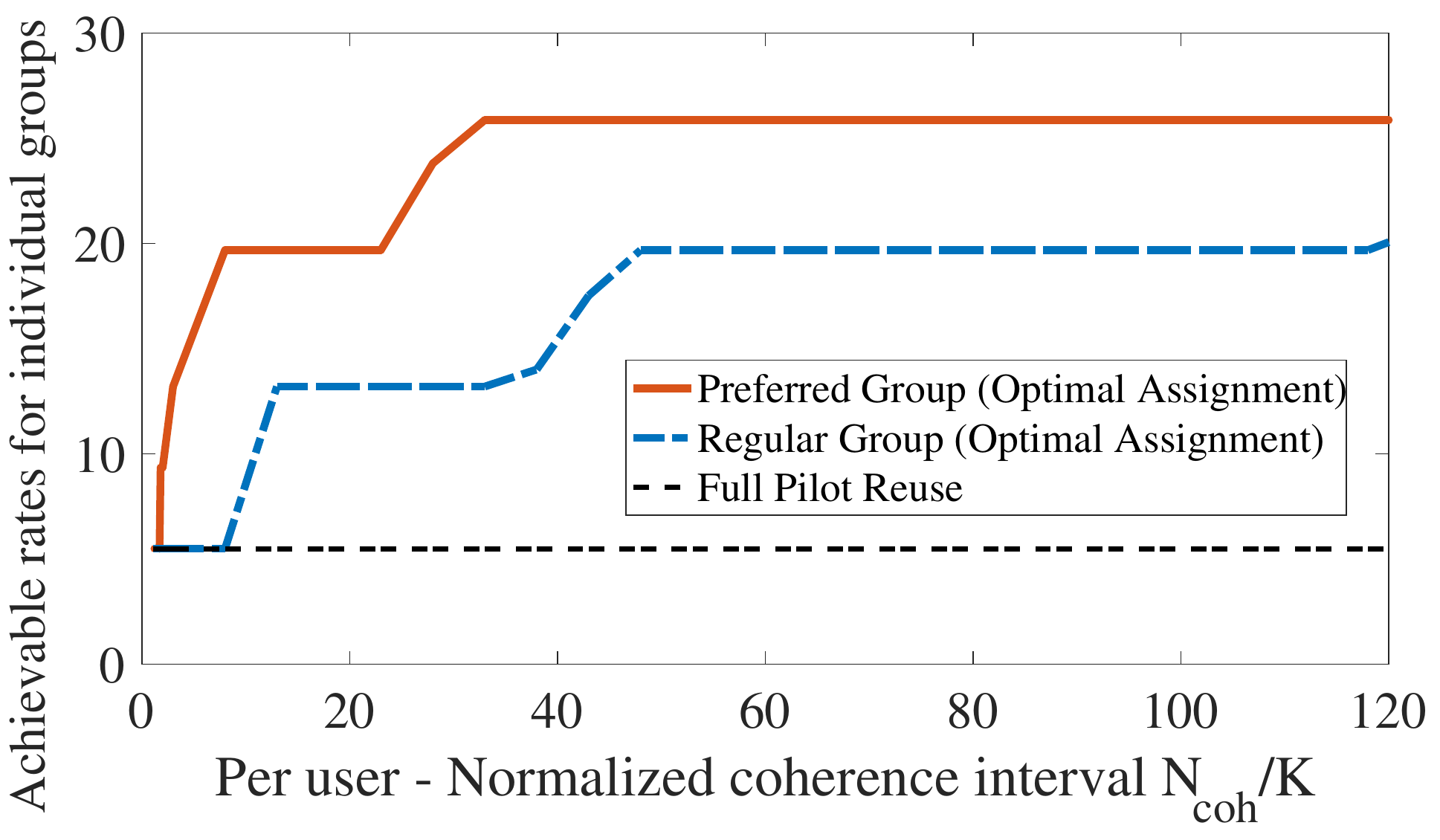}\label{Fig:alpha=0.2_omega=0.8_indiv_rates}}
	\quad \quad
	\subfloat[][$\alpha=0.2, \omega=0.6$]{\includegraphics[height=45mm]{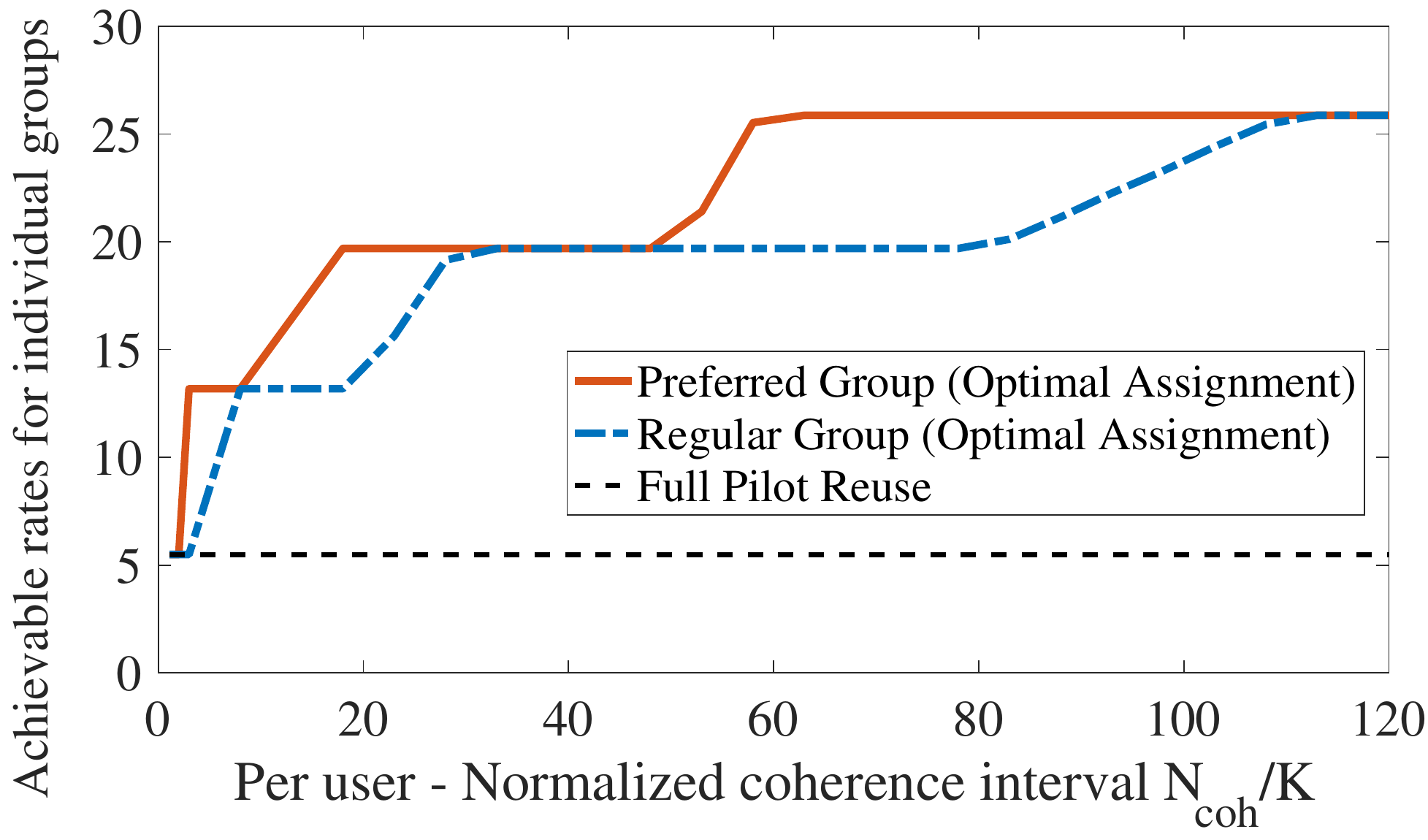}\label{Fig:alpha=0.2_omega=0.6_indiv_rates}}
		\quad \quad
	\subfloat[][$\alpha=0.4, \omega=0.8$]{\includegraphics[height=45mm]{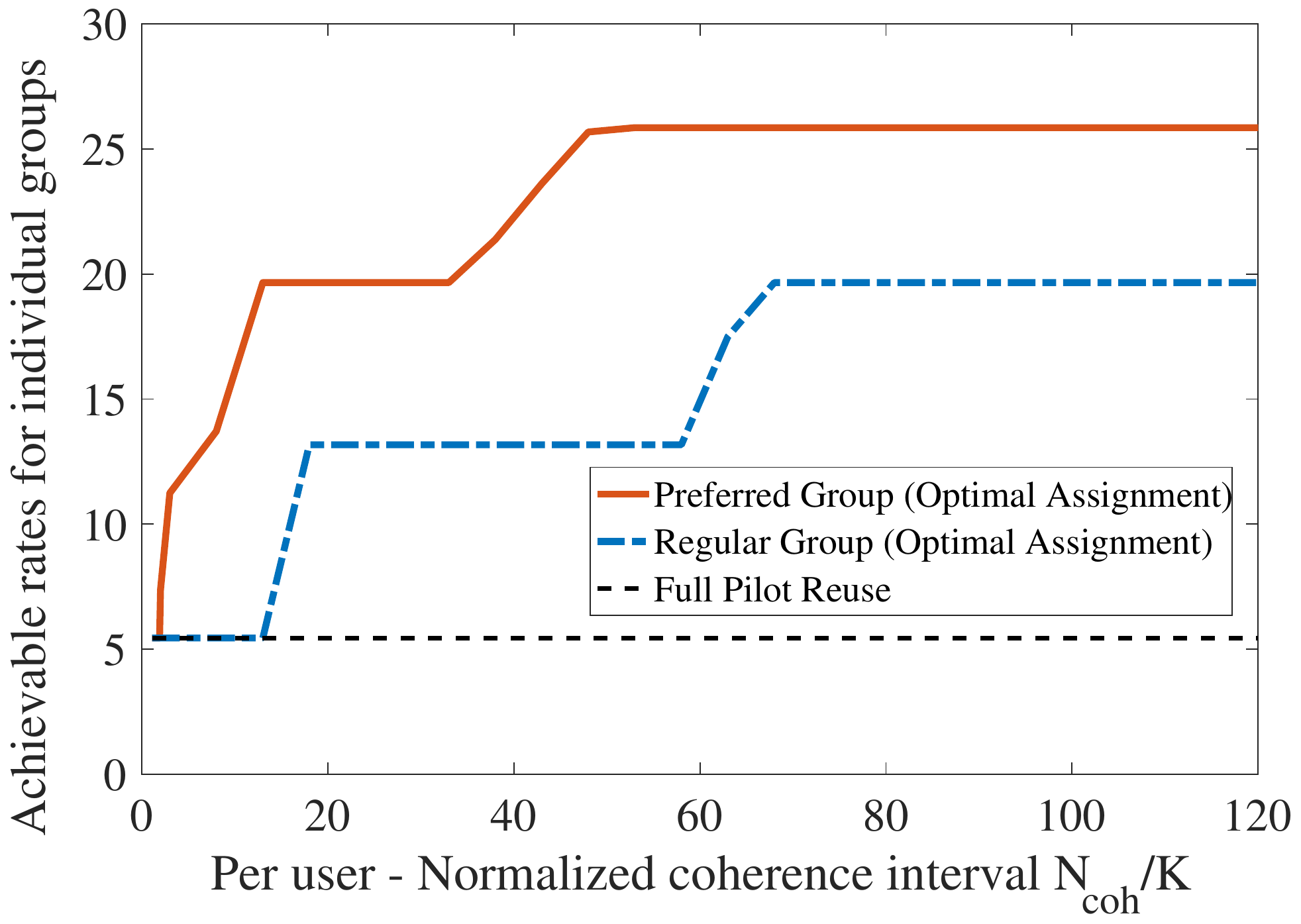}\label{Fig:alpha=0.4_omega=0.8_indiv_rates}}
	\caption{Per-user achievable rates of each priority group for optimal/conventional pilot assignments ($L=81,K=10$)}
	\label{blahblah}
\end{figure}


Now, we analyze the achievable rate averaged over each priority group. Figs. \ref{Fig:alpha=0.2_omega=0.8_indiv_rates}, \ref{Fig:alpha=0.2_omega=0.6_indiv_rates} and \ref{Fig:alpha=0.4_omega=0.8_indiv_rates} illustrate per-user achievable rates in the case of using conventional assignment versus optimal assignment, with different $\alpha$ and $\omega$ settings. 
In the case of conventional assignment, every user has a constant achievable rate value irrespective of $N_{coh}/K$, since the assignment rule is fixed so that the effect of pilot contamination does not change.
However, in the case of choosing optimal assignment specific to given $N_{coh}/K$ to maximize the net-WSR, each priority group can enjoy a higher achievable rate by utilizing extra pilots and spreading pilot-sharing users, which reduces the pilot contamination effect.
The allocation of extra pilots to the preferred group or the regular group depends on system setting: $\alpha$ and $\omega$.

Consider a case where $\alpha$ is fixed to $0.2$ and compare achievable rate curves for $\omega=0.8$ and $0.6$ (Figs. \ref{Fig:alpha=0.2_omega=0.8_indiv_rates} and \ref{Fig:alpha=0.2_omega=0.6_indiv_rates}, respectively).
We can observe that a larger $\omega$ value implies a wider performance gap between the preferred group and the regular group. This can be explained as follows. As we have a higher $\omega$ value, the influence of preferred group performance to the cost function (net-WSR) increases, so that it is more preferable to use extra pilots for $1^{st}$ priority group rather than $2^{nd}$ priority group, in order to maximize the net-WSR (note that the performance gap converges to zero, since the individual rates cannot exceed $C_{\log_3 L - 1}$ as mentioned in Section III-A).
Now, consider a scenario where $\omega$ is fixed to $0.8$ and compare the achievable rate curves for $\alpha=0.2$ and $0.4$ (Figs. \ref{Fig:alpha=0.2_omega=0.8_indiv_rates} and \ref{Fig:alpha=0.4_omega=0.8_indiv_rates}, respectively).
We can observe that as $\alpha$ increases, we have more users under $1^{st}$ priority, which requires more resources (coherence time) to boost up the achievable rate of the preferred group. Consequently, the achievable rate of the regular group increases slowly as $\alpha$ increases, as shown in the figures.


\section{Further Comments}\label{Section:Further Comments}

\subsection{Comments on Multiple Priority Groups}

The scenario of having multiple priority groups (greater than two) is considered in this subsection. Based on the mathematical analysis and simulation result obtained from the two priority group case, the result for the generalized setting can be anticipated. For the general case of $n$ priority groups, denote weight/ratio of $i^{th}$ priority group as $\omega_i$ and $\alpha_i$, respectively. The net-WSR maximization problem can be formulated as finding
\begin{align*}
\mathbf{p}_{opt}(N_{coh}) &= [\mathbf{p}_{opt}^{(1)}(N_{coh}), \cdots \mathbf{p}_{opt}^{(n)}(N_{coh})] \\
&\triangleq \underset{[\mathbf{p}_{1},\cdots, \mathbf{p}_{n}]}{\arg\max}\ C_{net,wsr}(\mathbf{p}_{1},\cdots,\mathbf{p}_{n},N_{coh})
\end{align*}
where $\mathbf{p}_{i} \in P_{L,\alpha_i K}$ for $i=1,\cdots,n$.
Here, 
\begin{align*} 
C_{net,wsr}(\mathbf{p}_{1},&\cdots,\mathbf{p}_{n}, N_{coh})\nonumber\\
&= \frac{N_{coh}- \sum_{i=1}^{n}N_{pil}(\mathbf{p}_i)}{N_{coh}} C_{wsr}(\mathbf{p}_{1},\cdots, \mathbf{p}_{n})
\end{align*}
and 
$C_{wsr}(\mathbf{p}_{1},\cdots, \mathbf{p}_{n}) = \sum_{i=1}^{n}\omega_i C_{sum}(\mathbf{p}_i).
$

\begin{table}
	\small
	\caption{Optimal pilot lengths for each priority group ($3$ priority group case - $L=27$, $K=10$, $\alpha_1 = \omega_3 =  0.2$, $\alpha_2 = \omega_2 = 0.3$, $\alpha_3 = \omega_1 = 0.5$)}
	\centering
	\label{Table: Pattern on 3 group case}
	\setlength\tabcolsep{1.5pt} 
\begin{tabular}{c|c|c|c|c|c|c|c|c|c|c|c|c|c|c}
	\hline 
	$\ \boldsymbol{T} \ $ & $10$ & $12$  & $14$ & $16$  & $18$ & $20$  & $22$ & $24$ & $\cdots$ & $30$ & $32$ & $34$ & $\cdots$ & $42$ \\ 
	\hline 
	$\ \boldsymbol{\rho_1}(\boldsymbol{T}) \ $& $2$ & $4$ & \multicolumn{8}{c|}{6} & $8$ & $10$ & $\cdots$ & $18$ \\ 
	\hline 
	$\ \boldsymbol{\rho_2}(\boldsymbol{T}) \ $& \multicolumn{3}{c|}{3} & $5$ & $7$ &  \multicolumn{9}{c}{9} \\ 
	\hline 
	$\ \boldsymbol{\rho_3}(\boldsymbol{T}) \ $& \multicolumn{6}{c|}{5} & $7$ & $9$ & $\cdots$ &  \multicolumn{5}{c}{15} \\ 
	\hline 
\end{tabular} 
\par\bigskip
\begin{tabular}{c|c|c|c|c|c|c|c|c}
	\hline 
	$\ \boldsymbol{T} \ $ & $44$ & $46$ & $\cdots$ & $60$ & $62$ & $64$ & $\cdots$ & $90$ \\ 
	\hline 
	$\ \boldsymbol{\rho_1}(\boldsymbol{T}) \ $& \multicolumn{8}{c}{18} \\ 
	\hline 
	$\ \boldsymbol{\rho_2}(\boldsymbol{T}) \ $& $11$ & $13$ & $\cdots$ &  \multicolumn{5}{c}{27} \\ 
	\hline 
	$\ \boldsymbol{\rho_3}(\boldsymbol{T}) \ $& \multicolumn{4}{c|}{15} & $17$ & $19$ & $\cdots$ & $45$ \\ 
	\hline 
\end{tabular}
\end{table}

Similar to the two priority group case, the problem can be divided into two sub-problems: 1) solving the optimal pilot assignment rule under the constraint of $\sum_{i=1}^{n}N_{pil}(\mathbf{p}_i) = T$, and 2) obtaining optimal $T$ value for a given $N_{coh}$. The solution to the $1^{st}$ sub-problem determines the optimal pilot length $\rho_1(T), \cdots, \rho_n(T)$ for each priority group (In the two-group case, the optimal lengths were $\rho(T)$ and $T-\rho(T)$). For a simple scenario as an example, the result of pilot resource allocation is expressed in Table \ref{Table: Pattern on 3 group case}. The system with $L=27$ cells and $K=10$ users per cell is assumed. The weight/ratio for three priority groups is fixed as $\alpha_1 = \omega_3 =  0.2$, $\alpha_2 = \omega_2 = 0.3$, $\alpha_3 = \omega_1 = 0.5$. 

Similar to Table \ref{Table: Pattern on optimal assignment_1} and \ref{Table: Pattern on optimal assignment_2}, the optimal pilot allocation to three priority groups has a certain pattern. As $T$ increases (i.e., additional pilots are allowed), one of three priority groups obtains extra pilots alternatively. Considering the WSR increment of allocating extra pilots to each priority group, optimal assignment chooses the priority group which maximizes the increment. Since the increment is proportional to $\omega_i$ value, the optimal rule chooses from a higher priority group to a lower priority group, sequentially. Therefore, in the general case of multiple priority groups, the pilot assignment rule for each priority group has certain characteristic which is easily obtained from the mathematical analysis on the two-group case.
Under the scenario of three priority groups, Fig. \ref{Fig:3Group_indiv_rates} illustrates the achievable rate averaged over each priority group. The system parameters are set to $L=27, K=10, \alpha_1=0.1, \alpha_2=0.4, \alpha_3=0.5, \omega_1 = 0.7, \omega_2 = 0.2, \omega_3 = 0.1$. Similar to the scenario with two priority groups, each priority group can enjoy a higher achievable rate by applying the net-WSR maximizing pilot assignment.

\begin{figure}[!t]
	\centering
	\includegraphics[height=50mm]{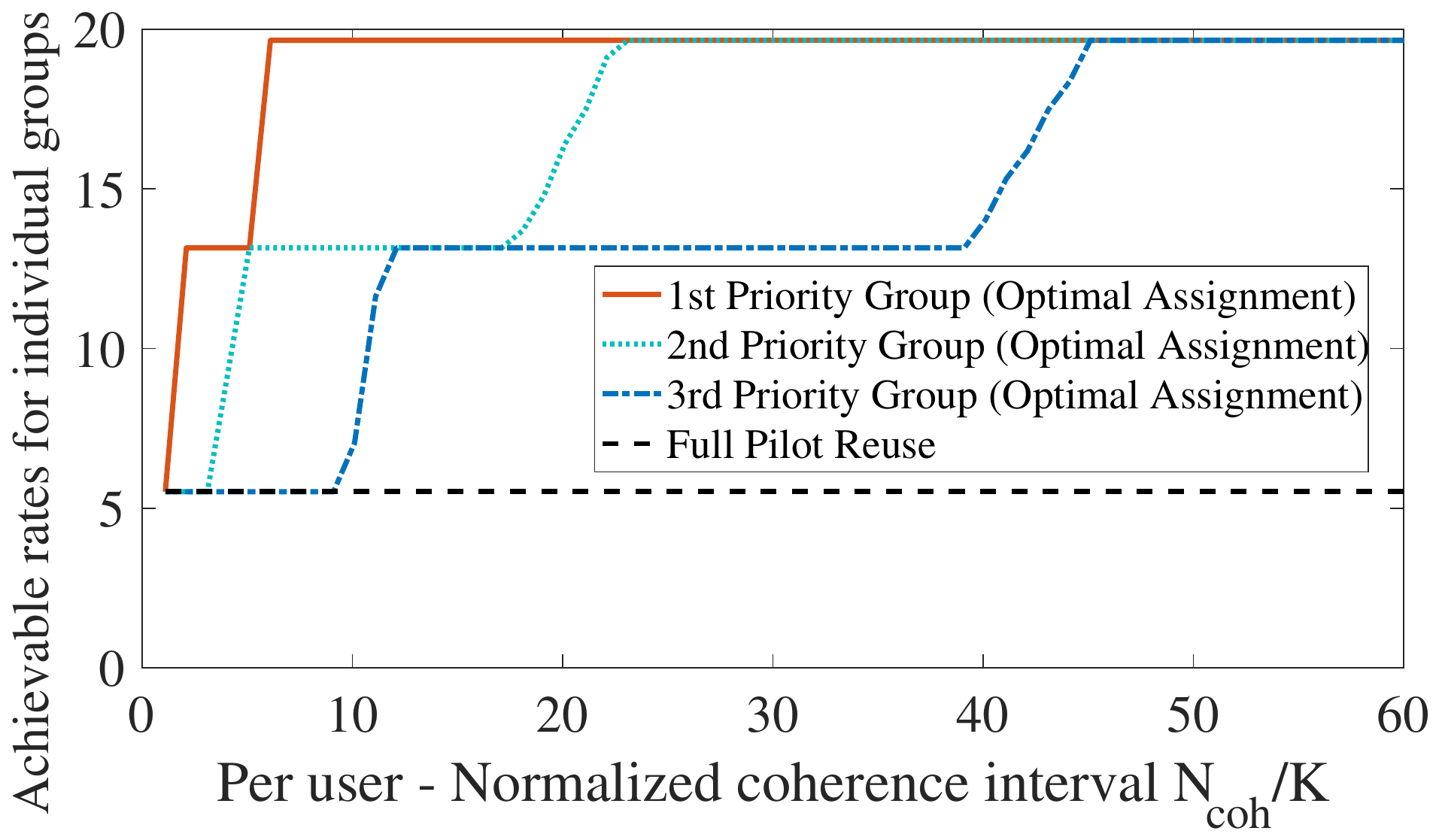}
	\caption{Per-user achievable rates of each priority group for optimal/conventional pilot assignments ($L=27,K=10, [\alpha_1, \alpha_2, \alpha_3]=[0.1,0.4,0.5], [\omega_1, \omega_2, \omega_3]=[0.7,0.2,0.1]$)}
	\label{Fig:3Group_indiv_rates}
\end{figure}

\subsection{Comments on finite antenna elements case}

This paper considers the scenario of having an infinite number of BS antennas in the mathematical analysis. However, a finite number of antennas will be deployed at each BS in practice, so that the performance of the suggested pilot assignment scheme needs to be confirmed for finite antenna elements. By Theorem 1 of \cite{bjornson2016massive}, the achievable rate of massive MIMO with pilot reuse factor $\beta$ can be obtained for finite $M$. 
Based on this result, the net-WSR maximizing optimal pilot assignment for finite $M$ is numerically obtained. 
Fig. \ref{Fig:FiniteM} illustrates the performance of optimal/conventional pilot assignment as a function of $M$. The receivers with maximal-ratio-combining (MRC) and zero-forcing-combining (ZFC) are compared, and the system parameters are assumed to be $L=81, K=10, \alpha=0.2, \omega = 0.7$ and $N_{coh}=200$. In both MRC and ZFC receivers, the optimal assignment has a performance gain compared to the conventional full pilot reuse, even in practical scenarios \cite{larsson2014massive, gao2015massive} of having $128$ antennas at each BS.
Under the assumption of using optimal assignment, the receivers with ZFC and MRC has a non-decreasing performance gap. This implies that an appropriate signal processing scheme can boost up the performance gain of optimal pilot assignment.

\begin{figure}[!t]
	\centering
	\includegraphics[height=50mm]{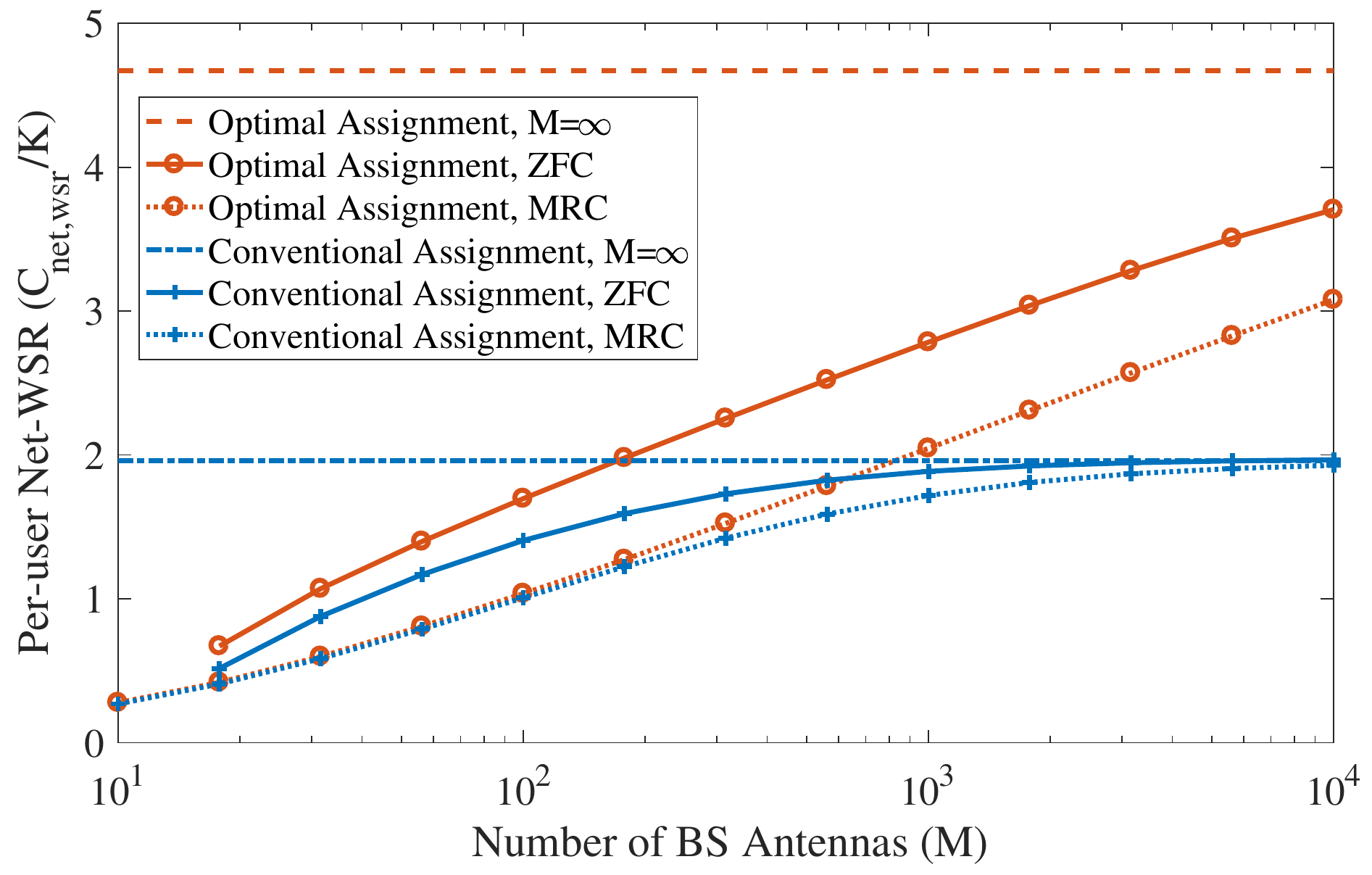}
	\caption{Per-user Net-WSR versus $M$ for optimal/conventional pilot assignments}
	\label{Fig:FiniteM}
\end{figure}

\section{Conclusion}\label{Section:Conclusion}

In a massive MIMO system with users grouped by priority, the optimal way of assigning pilots and the optimal portion of pilot training time to maximize the net-WSR have been found in a closed-form solution. As the available time slot $N_{coh}$ increases, the optimal strategy allows more time on pilot transmission, the allocation of extra pilot training time to either priority group depends on the weight $\omega$ and portion $\alpha$. Compared to the net sum-rate-maximization problem, our closed-form solution has better performance than conventional assignment for a wider range of $N_{coh}/K$ values, which means that our optimal solution can be applied to various practical channel scenarios. Simulation results of individual group achievable rates show that both groups can guarantee higher rates than conventional full pilot reuse, while a performance gap between the preferred/regular groups exists due to the different weights. The generalized scenario with multiple priority groups (greater than two) is also analyzed based on the insights obtained from the two-group case.

\appendices
\section{Proofs of Lemma 1, Lemma 2}\label{Appendix:Proofs of Lemmas 1 and 2}

Given the total pilot length $T$, it can be decomposed into two parts: pilot length $t$ for $1^{st}$ priority group, and pilot length $T-t$ for $2^{nd}$ priority group (as in the Table \ref{Table:Pilot length pairs}). Note that under the constraint on the pilot length $t$ for $1^{st}$ group, we already know that $\mathbf{p}'_{opt,K_1}(t)$ maximizes $C_{sum}$ from (\ref{Sohn2015:Thm1}). Similar results hold for $2^{nd}$ group. Therefore, when $t$ is fixed, we can obtain that $\bm{[}\: \mathbf{p}'_{opt, K_{1}}(t),\mathbf{p}'_{opt, K_{2}}(T-t) \:\bm{]}$ maximizes $C_{wsr}$.
However, a given $T$ can be decomposed into many possible $(t,T-t)$ pairs (note that the possible $t$ values are specified as $S_0(T)$).
All we need to find is optimal $t$, the pilot length for $1^{st}$ group.

Let $f(t)=C_{wsr}(\mathbf{p}'_{opt, K_{1}}(t), \mathbf{p}'_{opt,K_{2}}(T-t))$.
Then, using Collorary 1 of \cite{sohn2015pilots} and approximating $C_{i}$ as a linear function ($C_{i+1}-C_{i} \simeq 6$), we have
\begin{equation}\label{Appendix_A: equation1}
f(t) - f(t+2) = 6\: \omega\: 3^{-\chi(t, K_1)} (g_T(t)-1).
\end{equation}
Therefore, WSR comparison for two consecutive candidates $[\mathbf{p}'_{opt}(t, K_{1}), \mathbf{p}'_{opt}(T-t, K_{2})]$ and $[\mathbf{p}'_{opt}(t+2, K_{1}), \mathbf{p}'_{opt}(T-t-2, K_{2})]$ can be simplified by checking the sign of $g_T(t) - 1$ for $t \in S_1(T)$. 

\subsection{Proof of Lemma 1}

For $T$ values with $S_1(T)=\emptyset$, $S_0(T)$ has a single element $B(T)$, so that $[\mathbf{p}'_{opt, K_{1}}(B(T)),\\ \mathbf{p}'_{opt,K_{2}}(T-B(T))]$ is optimal trivially. The following proof deals with $T$ values with $S_1(T)\neq \emptyset$. 
Note that when $\log_3 \frac{\omega}{1-\omega} \notin \mathbb{Z}$, we have no $t \in S_{1}(T)$ such that $g_T(t)=1$. Using (\ref{Appendix_A: equation1}) and the fact that $g_T(t)$ is monotone increasing in $t$, we obtain optimal $t_{0}$ which maximizes $f(t)$ among $t\in S_0(T)$. The proof is divided into three cases.

First, if $g_T(B(T)) > 1$, we have $g_T(t)>1$ (i.e., $f(t) > f(t+2)$)  $\forall t \in S_{1}(T)$, so that $B(T)$ maximizes $f(t)$ among $t\in S_0(T)$. Second, if $g_T(F(T)-2) < 1$, we have $g_T(t)<1$ (i.e., $f(t) < f(t+2)$)  $\forall t \in S_{1}(T)$, so that $F(T)$ maximizes $f(t)$ among $t\in S_0(T)$. Finally, for the else case (i.e., $g_T(B(T)) \leq 1 \leq g_T(F(T)-2)$), 
denote $\xi(T)=min \{ t\in S_1(T)  \: : \: g_T(t) \geq 1 \}$.
Then, we have $g_T(t)<1$ (i.e., $f(t) < f(t+2)$) for $t = B(T), B(T) + 2, \cdots, \xi(T) - 2$ and $g_T(t)>1$ (i.e., $f(t) > f(t+2)$) for $ t = \xi(T), \xi(T) + 2, \cdots, F(T)$, so that $\xi(T)$ maximizes $f(t)$ among $t\in S_0(T)$.

\subsection{Proof of Lemma 2}

For the cases of $S_1(T)=\emptyset$, $g_T(B(T)) > 1$, or $g_T(F(T)-2) < 1$, similar approach to the proof of Lemma 1 can be applied. All we need to prove is the case of $g_T(B(T)) \leq 1 \leq g_T(F(T)-2)$, which implies $\exists t \in S_{1}(T)$ such that $g_T(t)=1$. For this case, denote $\xi(T) = min \{ t\in S_1(T)  \: : \: g_T(t) \geq 1 \}$ and $\eta(T)=max \{ t\in S_1(T)  \: : \: g_T(t) \leq 1 \} + 2$.
Based on (\ref{Appendix_A: equation1}), we have three equations: $f(t) < f(t+2)$ for $t \leq \xi(T)$, $f(t) = f(t+2)$ for $\xi(T) + 2 \leq t \leq \eta(T) - 2$, and $f(t) > f(t+2)$ for $t \geq \eta(T)$. Therefore, we can conclude that $\forall t \in \{\xi(T), \xi(T) + 2, \cdots, \eta(T) \}$ maximizes $f(t)$ among $t\in S_1(T)$.

\section{Proofs of Propositions 1 and 2}\label{Appendix:Proofs of Propositions 1 and 2}

\subsection{Proof of Proposition 1}
Here, we prove that the expression for $\rho(T)$ in (\ref{def:rho_T}) coincides with (\ref{def:rho_T_alternative}) and (\ref{def:rho_T_alternative2}).
Note that $S_1(T)=\emptyset$ (i.e., $B(T) = F(T)$) holds if and only if $T=K$ or $T=LK/3$, since we consider the $L \geq 9$ case. 
We begin with the $S_1(T)=\emptyset$ case, and proceed with the proof for the $S_1(T)\neq \emptyset$ case (i.e., when $T= K+2, K+4, \cdots, LK/3 -2$).

When $T=K$, $\rho(T)=B(T)=K_1$ from (\ref{def:rho_T}), which coincides with $\rho(T)=T-K_2 = K - K_2 = K_1$ in (\ref{def:rho_T_alternative}) and (\ref{def:rho_T_alternative2}). When $T=LK/3$, $\rho(T)=B(T)=LK_1/3$ from (\ref{def:rho_T}), which coincides with $\rho(T)=LK_1/3$ in (\ref{def:rho_T_alternative}) and (\ref{def:rho_T_alternative2}). Now we proceed to the case of $K < T < LK/3$.



\textbf{Case A:} if $3^s \geq \frac{L}{3}$

Expressions (\ref{def:rho_T}) and (\ref{def:rho_T_alternative}) are compared.
The proof is divided into three cases.

\textbf{Case A-1} (for $K < T \leq K_2 + \frac{LK_1}{3}$)

$F(T)=T-K_2$ from the definition of $F(T)$. Therefore, all we need to prove is $\rho(T)=F(T)$ for this $T$ interval. 
Note that $g_T(F(T)-2)=g_T(T-K_2-2)=\frac{1-w}{w} 3^{\chi(T-K_2-2, K_1)}$.
Since $\chi(N_{p0},K)$ is a monotonically increasing function of $N_{p0}$, we have
$g_T(F(T)-2) \leq \frac{1-\omega}{\omega} 3^{\chi(LK_1/3 - 2, K_1)} < 3^{-(s-1)}\frac{L}{9}\leq 1$
holds for $K < T \leq K_2 + \frac{LK_1}{3}$.
Therefore, from (\ref{def:rho_T}), $g_T(F(T)-2)<1$ implies $\rho(T)=F(T)$, which completes the proof.

\textbf{Case A-2} (for $K_2 + \frac{LK_1}{3} < T < \frac{LK}{3}$)

From case A-1, we have
$g_T(F(T)-2) < 1$ for $T=K_2 + LK_1/3$.
Also, from the definition of $F(T)$, we have $F(T+2)=F(T)=LK_1/3$ for $K_2 + LK_1/3 \leq T < LK/3$.
Thus, $g_{T+2}(F(T+2)-2)=g_{T+2}(F(T)-2)\leq g_{T}(F(T)-2)<1$. In summary, $g_{T}(F(T)-2)<1$ holds for $K_2 + \frac{LK_1}{3} < T < \frac{LK}{3}$, which implies $\rho(T)=F(T)$ from (\ref{def:rho_T}). Using $F(T) = LK_1/3$ for given range of $T$, we can confirm (\ref{def:rho_T}) coincides with  (\ref{def:rho_T_alternative}).


\textbf{Case B:} if $3^s < \frac{L}{3}$

Expressions (\ref{def:rho_T}) and (\ref{def:rho_T_alternative2}) are compared.
The proof is divided into three cases.

\textbf{Case B-1} (for $K < T \leq K_2 + 3^{s}K_1$)

Taking similar approach to case A-1 results in $\rho(T)=F(T)=T-K_2$ for these T values.

\textbf{Case B-2} (for $K_2 + 3^{s}K_1 < T < \frac{LK_1}{3} + \frac{LK_2}{3^{s+1}} $)

Denote $\xi(T) = min \{ t\in S_1(T)  \:\:|\enspace g_T(t) \geq 1 \}$.
The proof is divided into two parts: we first prove $g_T(B(T)) \leq 1 \leq g_T(F(T)-2)$ and then prove $\xi(T)=\phi(T)$ for a given $T$ range.
First, recall $B(T)=max(K_1, T-LK_2/3)$. If $B(T)=K_1$, then 
\begin{align*}
g_T(B(T)) &= 3^{-\chi(T-K_1-2, K_2)}\frac{1-\omega}{\omega}\nonumber\\
&< 3^{-\chi(K_2 + 3^s K_1 - K_1 - 2, K_2)}\frac{1-\omega}{\omega} \leq \frac{1-\omega}{\omega} \leq 1
\end{align*}
since $\chi$ is a nonnegative-valued function. If $B(T)=T-\frac{LK_2}{3}$, then 
\begin{equation*}
g_T(B(T)) = 3^{\chi(T-\frac{LK_2}{3}, K_1)-\chi(\frac{LK_2}{3}-2, K_2)}\frac{1-\omega}{\omega}
\leq 1
\end{equation*}
since $T<LK/3$. 
Similarly, recall $F(T)=min(T-K_2, \frac{LK_1}{3})$. If $F(T)=T- K_2$, then
\begin{align*}
g_T&(F(T)-2) = 3^{\chi(T-K_2-2, K_1)-\chi(K_2, K_2)}\frac{1-\omega}{\omega}\nonumber\\
&\geq 3^{\chi(3^s K_1, K_1) - \chi(K_2, K_2)}\frac{1-\omega}{\omega} = 3^{s} \frac{1-\omega}{\omega} \geq 1
\end{align*}
since $T \geq K_2 + 3^s K_1 +2$. If $F(T)=\frac{LK_1}{3}$, then
\begin{align*}
g_T&(F(T)-2) = 3^{\chi(\frac{LK_1}{3}-2, K_1)-\chi(T-\frac{LK_1}{3}, K_2)}\frac{1-\omega}{\omega}\\
&= 3^{log_3 L - 2} 3^{-\chi(T-\frac{LK_1}{3}, K_2)}\frac{1-\omega}{\omega} \nonumber\\
&\geq 3^{log_3 L - 2} 3^{-(log_3 L - 2 - s)}\frac{1-\omega}{\omega} = 3^s \frac{1-\omega}{\omega} \geq 1
\end{align*}
since $T \leq \frac{LK_1}{3} + \frac{LK_2}{3^{s+1}}-2$. Thus, we have $g_T(B(T)) \leq 1 \leq g_T(F(T)-2)$ for given $T$ range.

Now, we just need to prove $\xi(T)=\phi(T)$.
We start with proving $\xi(T)=3^{V(T)+s-1}K_1$ for $T \leq 3^{V(T)+s-1}K_1 +  3^{V(T)} K_2$. 
First,
\begin{align*}
&g_T(3^{V(T)+s-1}K_1 - 2)\\
&= 3^{\chi(3^{V(T)+s-1}K_1-2, K_1) - \chi(T-3^{V(T)+s-1}K_1, K_2)}\frac{1-\omega}{\omega}\\
&= 3^{V(T)+s-2 - \chi(T-3^{V(T)+s-1}K_1, K_2)}\frac{1-\omega}{\omega}\\
&< 3^{V(T)+s-2 - \chi(3^{V(T)-1}K_2, K_2)}\frac{1-\omega}{\omega}\\
&= 3^{V(T)+s-2 -V(T)+1}\frac{1-\omega}{\omega}=3^{s-1}\frac{1-\omega}{\omega}<1
\end{align*}
where the first inequality is from $T > 3^{V(T)+s-1}K_1 + 3^{V(T)-1}K_2$ by the definition of $V(T)$. Moreover,
\begin{align*}
&g_T(3^{V(T)+s-1}K_1)\\ 
&= 3^{\chi(3^{V(T)+s-1}K_1, K_1) - \chi(T-3^{V(T)+s-1}K_1-2, K_2)}\frac{1-\omega}{\omega}\\
&= 3^{V(T)+s-1 - \chi(T-3^{V(T)+s-1}K_1-2, K_2)}\frac{1-\omega}{\omega}\\
&\geq 3^{V(T)+s-1 - \chi(3^{V(T)}K_2-2, K_2)}\frac{1-\omega}{\omega} = 3^{s}\frac{1-\omega}{\omega}\geq 1
\end{align*}
where the first inequality is from $T \leq 3^{V(T)+s-1}K_1 +  3^{V(T)} K_2$.
Therefore, $\xi(T)=min \{ t\in S_1(T)  \:  : \:  g_T(t) \geq 1 \}=3^{V(T)+s-1}K_1$.

The next step is to prove $\xi(T)=T - 3^{V(T)}K_2$ for $T > 3^{V(T)+s-1}K_1 +  3^{V(T)} K_2$. First, 
\begin{align*}
&g_T(T - 3^{V(T)}K_2 - 2) \\
&= 3^{\chi(T - 3^{V(T)}K_2 - 2, K_1) - \chi(3^{V(T)}K_2, K_2)}\frac{1-\omega}{\omega}\\
&\leq 3^{\chi(3^{V(T)+s}K_1-2, K_1) - V(T)}\frac{1-\omega}{\omega} \leq 3^{s-1}\frac{1-\omega}{\omega} < 1
\end{align*}
where the first inequality is from $T \leq 3^{V(T)+s}K_1 + 3^{V(T)}K_2$ by the definition of $V(T)$.
Also, 
\begin{align*}
&g_T(T-3^{V(T)}K_2)\\
&= 3^{\chi(T-3^{V(T)}K_2, K_1) - \chi(3^{V(T)}K_2-2, K_2)}\frac{1-\omega}{\omega}\\
&= 3^{\chi(T-3^{V(T)}K_2, K_1) -V(T)+1} \frac{1-\omega}{\omega}\\
&> 3^{\chi(3^{V(T)+s-1}K_1, K_1)} 3^{-V(T)+1}\frac{1-\omega}{\omega}\\
&= 3^{V(T)+s-1-V(T)+1}\frac{1-\omega}{\omega} = 3^{s}\frac{1-\omega}{\omega}\geq 1
\end{align*}
where the first inequality is from $T > 3^{V(T)+s-1}K_1 +  3^{V(T)} K_2$.
Therefore, $\xi(T)=min \{ t\in S_1(T)  \: : \: g_T(t) \geq 1 \}=T - 3^{V(T)}K_2$. In summary, $\xi(T)=\phi(T)$ for a given range of $T$.

\textbf{Case B-3} (for $\frac{LK_1}{3} + \frac{LK_2}{3^{s+1}} \leq T < \frac{LK}{3}$)

From the definition of $F(T)$, we have $F(T)=\frac{LK_1}{3}$. Therefore,
$g_T(F(T)-2)=3^{\chi(\frac{LK_1}{3}-2, K_1) - \chi(T-\frac{LK_1}{3}, K_2)}\frac{1-\omega}{\omega}\leq 3^{s-1}\frac{1-\omega}{\omega} < 1.
$ From (\ref{def:rho_T}), we have $\rho(T)=F(T)=\frac{LK_1}{3}$.

\subsection{Proof of Proposition 2}

Consider $3^s \geq \frac{L}{3}$ case. 
From Proposition 1, we have $\\\rho(T+2)=\rho(T)+2$ for $K \leq T < K_2 + \frac{LK_1}{3}$ and $\rho(T+2)=\rho(T)=\frac{LK_1}{3}$ for $K_2 + \frac{LK_1}{3} \leq T < \frac{LK}{3}$.
If $3^s < \frac{L}{3}$, from the proposition 1, we have $\rho(T+2)=\rho(T)+2$ for $K \leq T < K_2 + 3^s K_1$ and $\rho(T+2)=\rho(T)=\frac{LK_1}{3}$ for $\frac{LK_1}{3} + \frac{LK_2}{3^{s+1}} \leq T < \frac{LK}{3}$.
In case of $K_2 + 3^s K_1 \leq T < \frac{LK_1}{3} + \frac{LK_2}{3^{s+1}}$, we have $\rho(T+2)=\rho(T)$ for $T < 3^{V(T)+s-1}K_1 + 3^{V(T)}K_2$ and $\rho(T+2)=\rho(T)+2$ for $T \geq 3^{V(T)+s-1}K_1 + 3^{V(T)}K_2$.
Therefore, for every case, we have either $\rho(T+2)=\rho(T)+2$ or $\rho(T+2)=\rho(T)$.

\section{Proof of Corollary 1}\label{Appendix:Proof of Corollary 1}
We start with the case of $T \geq \frac{LK_1}{3} + \max\{\frac{L}{3^{s+1}}, 1\} \: K_2$. From Proposition 1, we have $\rho(T+2) = \rho(T)=\frac{LK_1}{3}$ in this $T$ range. Therefore, from Theorem 1, we have $\bm{[}\: \mathbf{p}'_{opt, K_{1}}(\frac{LK_1}{3}),\mathbf{p}'_{opt, K_{2}}(T-\frac{LK_1}{3}) \:\bm{]} \in \widetilde P_{opt}(T)$, which result in $\bar{C}_{wsr}(T)=\omega C_{sum}(\mathbf{p}'_{opt, K_{1}}(\frac{LK_1}{3})) + (1-\omega) C_{sum}(\mathbf{p}'_{opt, K_{2}}(T-\frac{LK_1}{3}))$. Finally, we obtain $\bar{C}_{wsr}(T+2) - \bar{C}_{wsr}(T) =
(1-\omega) C_{sum}(\mathbf{p}'_{opt, K_{2}}(T+2-\frac{LK_1}{3})) - (1-\omega) C_{sum}(\mathbf{p}'_{opt, K_{2}}(T-\frac{LK_1}{3}))
=  (1-\omega) L 3^{-d_2} (C_{d_2 + 1} - C_{d_2} )$ where $d_2 = \chi(T- \rho(T), K_2)$.

In other cases ($T < \frac{LK_1}{3} + \max(\frac{L}{3^{s+1}}, 1) \: K_2$), we have either $\rho(T+2) = \rho(T)$ or $\rho(T+2) = \rho(T) + 2$ from Proposition 2. Denote $\mathbf{p}_{1}^{*} = \mathbf{p}'_{opt, K_{1}}(\rho(T))$ and $\mathbf{p}_{2}^{*} = \mathbf{p}'_{opt, K_{2}}(T+2 - \rho(T))$. Also, denote $\mathbf{p}_{1}^{**} = \mathbf{p}'_{opt, K_{1}}(\rho(T)+2)$ and $\mathbf{p}_{2}^{**} = \mathbf{p}'_{opt, K_{2}}(T - \rho(T))$. Then, we can confirm that $[\mathbf{p}_{1}^{*}, \mathbf{p}_{2}^{*}] \in \Theta(T+2)$ and $[\mathbf{p}_{1}^{**}, \mathbf{p}_{2}^{**}] \in \Theta(T+2)$. Since $\widetilde P_{opt}(T)$ chooses $[\mathbf{p}_{1}, \mathbf{p}_{2}] \in \Theta(T)$ which maximizes the WSR, we have $\bar{C}_{wsr}(T+2) = \max\{C_{wsr}(\mathbf{p}_{1}^{*}, \mathbf{p}_{2}^{*}), C_{wsr}(\mathbf{p}_{1}^{**}, \mathbf{p}_{2}^{**})\}$. Comparing this value with $\bar{C}_{wsr}(T) = C_{wsr}(\mathbf{p}'_{opt, K_{1}}(\rho(T)), \mathbf{p}'_{opt, K_{2}}(T-\rho(T)))$, we have $\bar{C}_{wsr}(T+2) = \bar{C}_{wsr}(T) + \delta_T$ where $\delta_T = max \{(1-\omega) L 3^{-d_2} (C_{d_2 + 1} - C_{d_2}),\omega L 3^{-d_1} (C_{d_1 + 1} - C_{d_1})\}$ (note that $d_1$ and $d_2$ are defined in the statement of Corollary 1).

\section{Proof of Theorem 2}\label{Appendix:Proof of Theorem 2}
We wish to find $\mathbf{p}_{opt}(N_{coh})$ that maximizes $C_{net,wsr}$ for any $N_{coh}$. Since $\bm{[}\: \mathbf{p}'_{opt, K_{1}}(\rho(T)),\mathbf{p}'_{opt, K_{2}}(T-\rho(T)) \:\bm{]}$ maximizes $C_{wsr}$ for the given $T$ constraint (from Theorem \ref{Theorem:Theorem1}), all we need to compare is $h_T(N_{coh})$ values for different $T \in \{K, K+2, \cdots, LK/3 \}$. Fig. \ref{Fig:Thm2 proof graphics} illustrates the graph of $h_T(N_{coh})$ for three consecutive $T$ values.

\begin{figure}[!t]
	\centering
   \includegraphics[height=30mm]{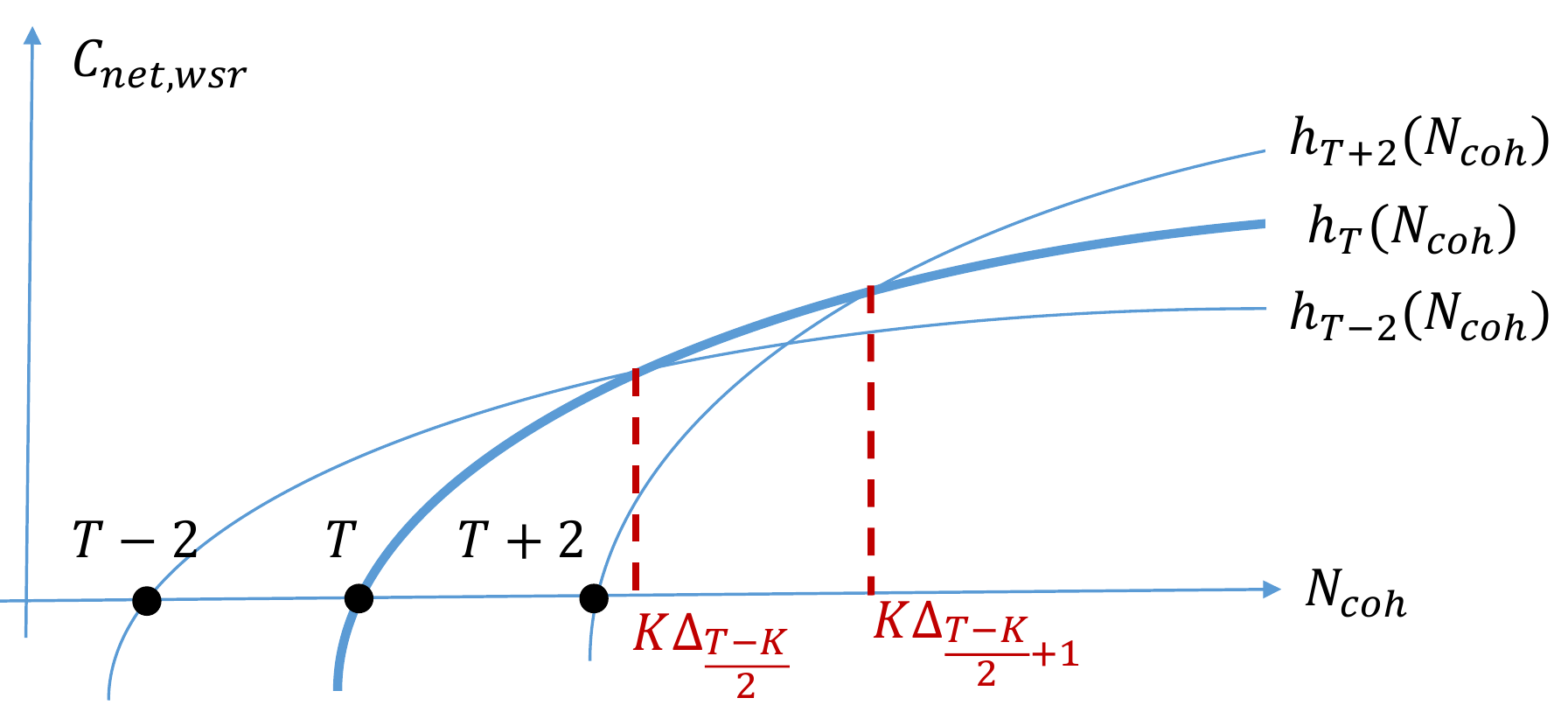}
    \caption{Graphical explanation for proof of Theorem 2}
    \label{Fig:Thm2 proof graphics}
\end{figure}

We start with checking the $N_{coh}$ value where $h_T(N_{coh})$ and $h_{T+2}(N_{coh})$ crosses. Based on the definition, $h_T(N_{coh}) = h_{T+2}(N_{coh})$ reduces to $N_{coh} = T + 2 + \frac{2\bar{C}_{wsr}(T)}{\delta_T}$. Similarly, $h_{T-2}(N_{coh})$ and $h_{T}(N_{coh})$ crosses at $N_{coh} = T + \frac{2\bar{C}_{wsr}(T-2)}{\delta_{T-2}}$. Therefore, $\bm{[}\: \mathbf{p}'_{opt, K_{1}}(\rho(T)),\mathbf{p}'_{opt, K_{2}}(T-\rho(T)) \:\bm{]}$ has maximum $C_{net,wsr}$ for $N_{coh} \in [T + \frac{2\bar{C}_{wsr}(T-2)}{\delta_{T-2}}, T + 2 + \frac{2\bar{C}_{wsr}(T)}{\delta_T})$. For general $T\in \{K, K+2, \cdots, LK/3 \}$ values, we used sequence $\Delta_n$ for the final statement.



%

\bibliographystyle{IEEEtran}
\bibliography{IEEEabrv,JSAC2017}

\begin{thebibliography}{10}
\providecommand{\url}[1]{#1}
\csname url@samestyle\endcsname
\providecommand{\newblock}{\relax}
\providecommand{\bibinfo}[2]{#2}
\providecommand{\BIBentrySTDinterwordspacing}{\spaceskip=0pt\relax}
\providecommand{\BIBentryALTinterwordstretchfactor}{4}
\providecommand{\BIBentryALTinterwordspacing}{\spaceskip=\fontdimen2\font plus
\BIBentryALTinterwordstretchfactor\fontdimen3\font minus
  \fontdimen4\font\relax}
\providecommand{\BIBforeignlanguage}[2]{{%
\expandafter\ifx\csname l@#1\endcsname\relax
\typeout{** WARNING: IEEEtran.bst: No hyphenation pattern has been}%
\typeout{** loaded for the language `#1'. Using the pattern for}%
\typeout{** the default language instead.}%
\else
\language=\csname l@#1\endcsname
\fi
#2}}
\providecommand{\BIBdecl}{\relax}
\BIBdecl

\bibitem{marzetta2006much}
T.~L. Marzetta, ``How much training is required for multiuser mimo?'' in
  \emph{2006 Fortieth Asilomar Conference on Signals, Systems and
  Computers}.\hskip 1em plus 0.5em minus 0.4em\relax IEEE, 2006, pp. 359--363.

\bibitem{marzetta2010noncooperative}
------, ``Noncooperative cellular wireless with unlimited numbers of base
  station antennas,'' \emph{IEEE Transactions on Wireless Communications},
  vol.~9, no.~11, pp. 3590--3600, 2010.

\bibitem{lu2014overview}
L.~Lu, G.~Y. Li, A.~L. Swindlehurst, A.~Ashikhmin, and R.~Zhang, ``An overview
  of massive mimo: Benefits and challenges,'' \emph{IEEE Journal of Selected
  Topics in Signal Processing}, vol.~8, no.~5, pp. 742--758, 2014.

\bibitem{larsson2014massive}
E.~G. Larsson, O.~Edfors, F.~Tufvesson, and T.~L. Marzetta, ``Massive mimo for
  next generation wireless systems,'' \emph{IEEE Communications Magazine},
  vol.~52, no.~2, pp. 186--195, 2014.

\bibitem{andrews2014will}
J.~G. Andrews, S.~Buzzi, W.~Choi, S.~V. Hanly, A.~Lozano, A.~C. Soong, and
  J.~C. Zhang, ``What will 5g be?'' \emph{IEEE Journal on Selected Areas in
  Communications}, vol.~32, no.~6, pp. 1065--1082, 2014.

\bibitem{elijah2015comprehensive}
O.~Elijah, C.~Y. Leow, T.~A. Rahman, S.~Nunoo, and S.~Z. Iliya, ``A
  comprehensive survey of pilot contamination in massive mimo-5g system,''
  \emph{IEEE Communications Surveys \& Tutorials}, vol.~18, no.~2, pp.
  905--923, 2015.

\bibitem{appaiah2010pilot}
K.~Appaiah, A.~Ashikhmin, and T.~L. Marzetta, ``Pilot contamination reduction
  in multi-user tdd systems,'' in \emph{2010 IEEE International Conference on
  Communications (ICC)}, pp. 1--5.

\bibitem{yin2013coordinated}
H.~Yin, D.~Gesbert, M.~Filippou, and Y.~Liu, ``A coordinated approach to
  channel estimation in large-scale multiple-antenna systems,'' \emph{IEEE
  Journal on Selected Areas in Communications}, vol.~31, no.~2, pp. 264--273,
  2013.

\bibitem{jose2011pilot}
J.~Jose, A.~Ashikhmin, T.~L. Marzetta, and S.~Vishwanath, ``Pilot contamination
  and precoding in multi-cell tdd systems,'' \emph{IEEE Transactions on
  Wireless Communications}, vol.~10, no.~8, pp. 2640--2651, 2011.

\bibitem{ashikhmin2012pilot}
A.~Ashikhmin and T.~Marzetta, ``Pilot contamination precoding in multi-cell
  large scale antenna systems,'' in \emph{Information Theory Proceedings
  (ISIT), 2012 IEEE International Symposium on}, pp. 1137--1141.

\bibitem{li2015multi}
X.~Li, E.~Bjornson, E.~G. Larsson, S.~Zhou, and J.~Wang, ``A multi-cell mmse
  detector for massive mimo systems and new large system analysis,'' in
  \emph{2015 IEEE Global Communications Conference}, pp. 1--6.

\bibitem{ngo2012performance}
H.~Q. Ngo, M.~Matthaiou, and E.~G. Larsson, ``Performance analysis of large
  scale mu-mimo with optimal linear receivers,'' in \emph{Communication
  Technologies Workshop (Swe-CTW), 2012 Swedish}.\hskip 1em plus 0.5em minus
  0.4em\relax IEEE, 2012, pp. 59--64.

\bibitem{guo2013performance}
K.~Guo and G.~Ascheid, ``Performance analysis of multi-cell mmse based
  receivers in mu-mimo systems with very large antenna arrays,'' in
  \emph{Wireless Communications and Networking Conference (WCNC), 2013 IEEE},
  pp. 3175--3179.

\bibitem{guo2014uplink}
K.~Guo, Y.~Guo, G.~Fodor, and G.~Ascheid, ``Uplink power control with mmse
  receiver in multi-cell mu-massive-mimo systems,'' in \emph{Communications
  (ICC), 2014 IEEE International Conference on}, pp. 5184--5190.

\bibitem{zhu2015smart}
X.~Zhu, Z.~Wang, L.~Dai, and C.~Qian, ``Smart pilot assignment for massive
  mimo,'' \emph{IEEE Communications Letters}, vol.~19, no.~9, pp. 1644--1647,
  2015.

\bibitem{nguyen2015resource}
T.~M. Nguyen, V.~N. Ha, and L.~B. Le, ``Resource allocation optimization in
  multi-user multi-cell massive mimo networks considering pilot
  contamination,'' \emph{IEEE Access}, vol.~3, pp. 1272--1287, 2015.

\bibitem{liu2015pilot}
B.~Liu, Y.~Cheng, and X.~Yuan, ``Pilot contamination elimination precoding in
  multi-cell massive mimo systems,'' in \emph{Personal, Indoor, and Mobile
  Radio Communications (PIMRC), 2015 IEEE 26th Annual International Symposium
  on}, pp. 320--325.

\bibitem{bjornson2016massive}
E.~Bj{\"o}rnson, E.~G. Larsson, and M.~Debbah, ``Massive mimo for maximal
  spectral efficiency: How many users and pilots should be allocated?''
  \emph{IEEE Transactions on Wireless Communications}, vol.~15, no.~2, pp.
  1293--1308, 2016.

\bibitem{saxena2015mitigating}
V.~Saxena, G.~Fodor, and E.~Karipidis, ``Mitigating pilot contamination by
  pilot reuse and power control schemes for massive mimo systems,'' in
  \emph{2015 IEEE 81st Vehicular Technology Conference (VTC Spring)}.\hskip 1em
  plus 0.5em minus 0.4em\relax IEEE, 2015, pp. 1--6.

\bibitem{sohn2015pilots}
J.~Y. Sohn, S.~W. Yoon, and J.~Moon, ``When pilots should not be reused across
  interfering cells in massive mimo,'' in \emph{2015 IEEE International
  Conference on Communication Workshop (ICCW)}.\hskip 1em plus 0.5em minus
  0.4em\relax IEEE, 2015, pp. 1257--1263.

\bibitem{gao2015massive}
X.~Gao, O.~Edfors, F.~Rusek, and F.~Tufvesson, ``Massive mimo performance
  evaluation based on measured propagation data,'' \emph{IEEE Transactions on
  Wireless Communications}, vol.~14, no.~7, pp. 3899--3911, 2015.

\end{thebibliography}

%

\begin{IEEEbiography}[{\includegraphics[width=1in,height=1.25in,clip,keepaspectratio]{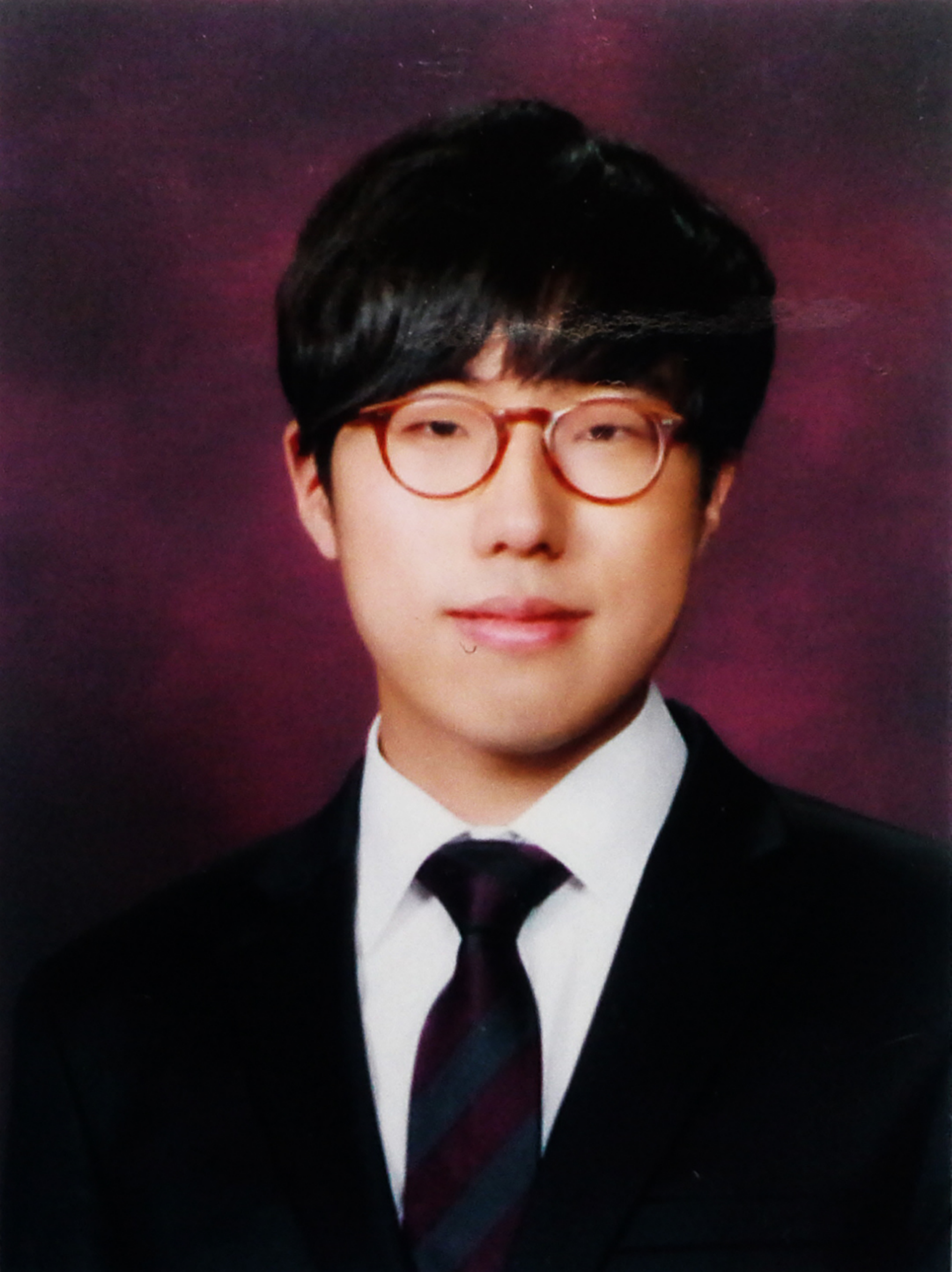}}]{Jy-yong Sohn}
received the B.S. and M.S. degrees in electrical engineering from the Korea Advanced Institute of Science and Technology (KAIST), Daejeon, Korea, in 2014 and 2016, respectively. He is currently pursuing the Ph.D. degree in KAIST. His research interests include massive MIMO effects on wireless multi cellular system and 5G Communications, with a current focus on distributed storage and network coding. \end{IEEEbiography}
\begin{IEEEbiography}[{\includegraphics[width=1in,height=1.25in,clip,keepaspectratio]{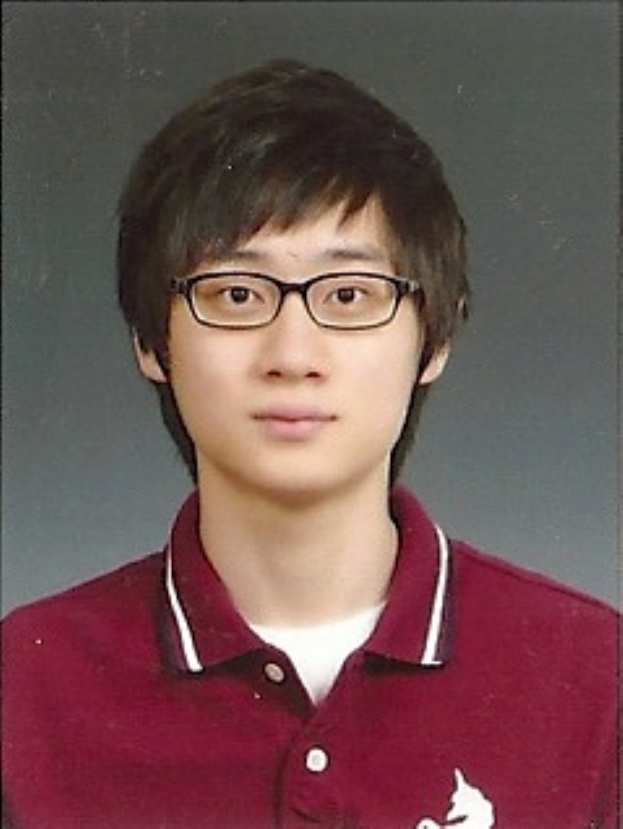}}]{Sung Whan Yoon}
received the B.S. and M.S. degrees in electrical engineering from the Korea Advanced Institute of Science and Technology (KAIST), Daejeon, Korea, in 2011 and 2013, respectively. He is currently pursuing the Ph.D degree in KAIST. His main research interests are in the field of coding and signal processing for wireless communication \& storage, especially massive MIMO, polar codes and distributed storage codes.
\end{IEEEbiography}
\begin{IEEEbiography}[{\includegraphics[width=1in,height=1.25in,clip,keepaspectratio]{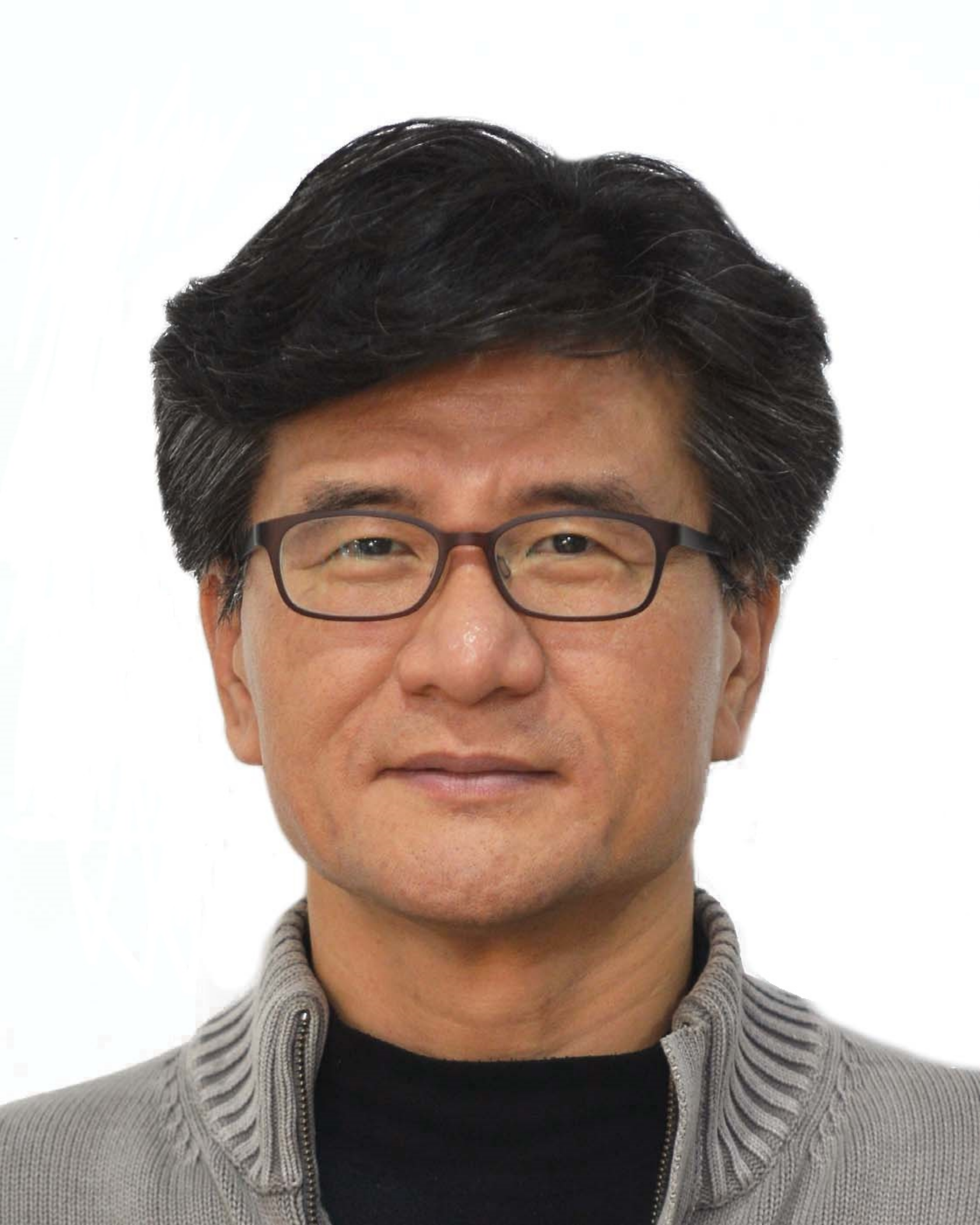}}]{Jaekyun Moon}
received the Ph.D degree in electrical and computer engineering at Carnegie Mellon University, Pittsburgh, Pa, USA. He is currently a Professor of electrical engineering at KAIST. From 1990 through early 2009, he was with the faculty of the Department of Electrical and Computer Engineering at the University of Minnesota, Twin Cities. He consulted as Chief Scientist for DSPG, Inc. from 2004 to 2007. He also worked as Chief Technology Officer at Link-A-Media Devices Corporation. His research interests are in the area of channel characterization, signal processing and coding for data storage and digital communication. Prof. Moon received the McKnight Land-Grant Professorship from the University of Minnesota. He received the IBM Faculty Development Awards as well as the IBM Partnership Awards. He was awarded the National Storage Industry Consortium (NSIC) Technical Achievement Award for the invention of the maximum transition run (MTR) code, a widely used error-control/modulation code in commercial storage systems. He served as Program Chair for the 1997 IEEE Magnetic Recording Conference. He is also Past Chair of the Signal Processing for Storage Technical Committee of the IEEE Communications Society, In 2001, he cofounded Bermai, Inc., a fabless semiconductor start-up, and served as founding President and CTO. He served as a guest editor for the 2001 IEEE JSAC issue on Signal Processing for High Density Recording. He also served as an Editor for IEEE TRANSACTIONS ON MAGNETICS in the area of signal processing and coding for 2001-2006. He is an IEEE Fellow.
\end{IEEEbiography}\vfill

\end{document}